\documentclass[reprint,superscriptaddress,amsmath,amssymb,aps,pra]{revtex4-2}

\usepackage{graphicx}
\usepackage{bm}
\usepackage[colorlinks=true,
  linkcolor=black,
  citecolor=blue,
  urlcolor =blue]{hyperref}
\usepackage{amsmath}
\usepackage{amssymb}
\usepackage{amsthm}
\usepackage{graphicx}
\usepackage{algorithmic}
\usepackage{mathrsfs}
\usepackage{braket}
\usepackage[dvipsnames]{xcolor}

\newtheorem{thm}{\protect\theoremname}
\theoremstyle{plain}
\newtheorem{lem}[thm]{\protect\lemmaname}
\theoremstyle{plain}

\theoremstyle{plain}
\newtheorem*{lem*}{\protect\lemmaname}
\theoremstyle{plain}

\theoremstyle{plain}

\usepackage[USenglish]{babel}
  \providecommand{\corollaryname}{Corollary}
  \providecommand{\lemmaname}{Lemma}
  \providecommand{\propositionname}{Proposition}
  \providecommand{\remarkname}{Remark}
\providecommand{\theoremname}{Theorem}

\usepackage[normalem]{ulem}

\newcommand{\Or}{\mathcal{O}}
\newcommand{\RR}{\mathbb{R}}

\newcommand{\wt}{\widetilde}
\newcommand{\Tr}{\mathrm{Tr}}
\newcommand{\tr}{\mathrm{tr}}
\newcommand{\dd}{\mathrm{d}}
\newcommand{\ZZ}{\mathbb{Z}}

\renewcommand{\ket}[1]{\ensuremath{\left|#1\right\rangle}}
\renewcommand{\bra}[1]{\ensuremath{\left\langle#1\right|}}

\begin{document}
\title{Learning conservation laws in unknown quantum dynamics}
\author{Yongtao Zhan}
\email{yzhan@caltech.edu}
\affiliation{Institute of Quantum Information and Matter, California Institute of Technology}
\affiliation{Division of Chemistry and Chemical Engineering, California Institute of Technology}

\author{Andreas Elben}
\email{aelben@caltech.edu}
\affiliation{Institute of Quantum Information and Matter, California Institute of Technology}
\affiliation{Walter Burke Institute for Theoretical Physics, California Institute of Technology}

 \author{Hsin-Yuan Huang}
 \email{hsinyuan@caltech.edu}
 \affiliation{Institute of Quantum Information and Matter, California Institute of Technology}
 \affiliation{Google Quantum AI}

 \author{Yu Tong}
\email{yutong@caltech.edu}
\affiliation{Institute of Quantum Information and Matter, California Institute of Technology}

\begin{abstract}
We present a learning algorithm for discovering conservation laws  given as sums of geometrically local observables in quantum dynamics. This includes conserved quantities that arise from local and global symmetries in closed and open quantum many-body systems.  
The algorithm combines the classical shadow formalism for estimating expectation values of observable and data analysis techniques based on singular value decompositions and robust polynomial interpolation to discover all such conservation laws in unknown quantum dynamics with rigorous performance guarantees. 
Our method can be directly realized in quantum experiments, which we illustrate with numerical simulations, using closed and open quantum system dynamics in a $\mathbb{Z}_2$-gauge theory and in many-body localized spin-chains.
\end{abstract}

\maketitle

\section{Introduction}

Machine learning (ML) is playing an increasingly important role in physical sciences \cite{carleo2019machine}. 
The ability of ML to recognize patterns in data greatly facilitates the data-driven approach to scientific research, where scientific discoveries are achieved by analyzing experimental data.
Recently, many works have been done to discover physical laws with ML models~\cite{doi:10.1126/sciadv.aay2631,10.3389/frai.2020.00025,PhysRevE.106.045307,lu2022discovering,doi:10.1126/science.1165620,doi:10.1126/science.1165893,PhysRevE.100.033311,doi:10.1073/pnas.1906995116,PhysRevLett.124.010508,PhysRevE.103.033303,PhysRevResearch.2.033499,NEURIPS2021_886ad506,kaiser2018discovering}. Following this line, several works attempt to learn conservation laws in classical mechanical systems~\cite{PhysRevLett.126.180604,PhysRevResearch.3.L042035}. The ML models in these works can successfully discover conserved quantities in simple classical systems, such as energy conservation, angular momentum conservation, and momentum conservation in two-body gravitational systems.

The conservation law, or the integral of motion, is also an important concept in quantum mechanics. There are usually many global conservation laws in quantum dynamics, such as the eigenstate projection operators, when the dynamics are governed by a Hamiltonian. However, these do not imply any special dynamical properties. In the quantum setting, 
the physically relevant conservation laws are the ones with locality structure \cite{Calabrese_2016}. Such local integrals of motions \cite{PhysRevLett.111.127201,PhysRevB.90.174202,ROS2015420,imbrie2017local} underlie, for instance, the absence of thermalization and transport in certain quantum systems -  in contrast to ergodic systems, which typically conserve only a few globally supported quantities such as the total energy or number of particles. Thus, local integrals of motions are central for our understanding of phenomena such as many-body localization (MBL) \cite{PhysRev.109.1492,PhysRevB.21.2366,PhysRevLett.95.206603,BASKO20061126,https://doi.org/10.1002/andp.201700169,doi:10.1146/annurev-conmatphys-031214-014726,RevModPhys.91.021001} and Hilbert space fragmentation \cite{moudgalya2022quantum}.

In this work, we consider a broad class of conservation laws with conserved quantities given by sums of geometrically local observables.
These conservation laws can be geometrically localized or have support across the entire system.
We propose an algorithm for discovering all such conservation laws in arbitrary quantum dynamical systems. 
We consider different types of quantum dynamics in which a quantum state $\rho(t)$ evolves with time, and its evolution is described by the von Neumann equation under a Hamiltonian $H$ or a Lindblad master equation under a Lindbladian $\mathcal{L}$.
Our algorithm is general enough to cover Hamiltonians that change over time, for instance, periodically, as in Floquet systems, and find observables that are conserved between periods.
The conservation laws we consider could be state-dependent, i.e., the observables are conserved for a certain subset of initial states. Such observables can be more difficult to find than those that are conserved for all states because they will be overlooked if only information from the Hamiltonian or the Lindbladian is used.

Our algorithm combines classical shadow formalism \cite{huang2020predicting} and several data analysis techniques.
The algorithm uses classical shadows to estimate the expectation values of many Pauli observables at multiple times during the quantum dynamics based on a limited number of randomized measurements \cite{huang2020predicting, elben2022randomized}.
After estimating the expectation values, the algorithm performs singular value decomposition (SVD) on a data matrix and gets the low dimensional manifold corresponding to the conservation laws.
We rigorously prove that the algorithm can efficiently learn all conserved quantities which are sums of geometrically local observables in quantum dynamics. Although some previous papers are using numerical techniques to find such conservation laws based on known Hamiltonian~\cite{PhysRevB.91.085425,Mierzejewski2015identify,PhysRevB.94.144208,PhysRevLett.126.180602,Bentsen2019integrable}, our method can be directly applied to experiments to find conservation laws in arbitrary unknown quantum dynamics.
Our method also comes with theoretical guarantees that the sample and computational complexities are both polynomial in the system size and precision.
Note that, for Hamiltonian dynamics, one could also obtain conservation laws by first learning the Hamiltonian \cite{wang2017experimental,evans2019scalable,gu2022practical,granade2012,hangleiter2021,wiebe2014a,wiebe2014b,yu2022,ZubidaYitzhakiEtAl2021optimal,GranadeFerrieWiebeCory2012robust,li2020hamiltonian,FrancaEtAl2022efficient,tang2021,pastori2022characterization,caro2022learning,huang2022learning,zhao2022supervised}. However, Hamiltonian learning protocols only work efficiently when there are some known sparsity or locality constraints on the Hamiltonian.
For quantum dynamics without such sparsity structure, prior works involve exponential sample complexity or classical post-processing cost. In comparison, our methods are not subject to these constraints.

We perform numerical experiments illustrating the learning of conserved quantities in closed and open system dynamics in a $\mathbb{Z}_2$ gauge theory, arising from local and global symmetries. In addition, we demonstrate the learning of local, approximate conservation laws in a one-dimensional XXZ-chain with local disorder. Here, a sharp increase in the number of local conserved quantities takes place at a certain disorder strength, which we successfully observe from the result of our algorithm.

Discovering conservation laws in quantum systems is a fundamental problem and has also been studied in prior works. To our knowledge, this is the first work to provide general rigorous guarantees  for learning and testing conservation laws in quantum experiments with unknown dynamics, by employing the randomized measurement toolbox \cite{elben2022randomized} and the classical shadow formalism \cite{huang2020predicting}.
While
Ref.~\cite{Bentsen2019integrable} proposes a similar algorithm to construct local conserved quantities in quantum dynamics, they use it in classical simulations of known time-independent Hamiltonian dynamics to study integrability.
Ref.~\cite{shtanko2023uncovering} conducted physical experiments to obtain conserved quantities that are supported in local regions, which excludes conserved quantities supported on the entire system, such as total magnetization. In contrast, our protocol can uncover conservation laws supported locally and globally.
Furthermore, \cite{shtanko2023uncovering} considers quantities that are conserved at discrete points in time, whereas we can also deal with the continuous-time scenario through robust polynomial interpolation.

\section{Algorithm Description}
\label{section2}

The goal of our algorithm is to find all conserved quantities that are linear combinations of Pauli operators supported on $k = \mathcal{O}(1)$ adjacent qubits, in an unknown quantum dynamical system with experimental feasible measurements. The Hamiltonian or Lindbladian that governs the quantum dynamics is completely unknown and need not to be local, but we can control it to evolve for a time of our choice and perform randomized single-qubit measurements.
Throughout this work, we will use $\rho(t)$ to denote the time-evolved state and $O(t)$ to denote the time-evolved observable in the Heisenberg picture.

\vspace{1em}
\noindent \textbf{Classical Shadows.} 
The classical shadow formalism \cite{huang2020predicting} was proposed to efficiently predict local observables with experimentally feasible randomized measurements~\cite{elben2022randomized}. To be specific, it was shown that one can predict $M$ arbitrary linear target function $\Tr(O_1 \rho),\cdots,\Tr(O_M \rho)$ up to additive error $\epsilon$ with only $\mathcal{O} (B \log (M) / \epsilon^2)$ measurements, where $B$ is the upper bound of the shadow norm defined in~\cite{huang2020predicting}. This result implies that a limited number of measurements are enough to predict the expectation values of a large number of observables. Making use of this property, we can predict all $k$-local Pauli observables with a limited number of random measurements.

The classical shadow formalism is summarized as follows: We approximate an $N$-qubit quantum state by performing randomized single-qubit Pauli measurements on $N_s$ copies of $\rho$. That is, we project each qubit to one of three Pauli basis $X,Y,Z$ and get a product state composed by the six basis states $\{|0\rangle,|1\rangle,|+\rangle,|-\rangle,|\mathrm{i}+\rangle,|\mathrm{i}-\rangle\}$. Performing one randomized measurement gives us such a product state, which can be stored in classical memory with an $N$-element array. After performing such measurements on $N_s$ copies of states, we get $N N_s$ single-qubit measurement results, which we can make use of to construct an approximation of the unknown state $\rho$:
\begin{equation}
\hat \rho=\frac{1}{N_s} \sum_{n_s=1}^{N_s} \hat \rho_1^{(n_s)} \otimes \cdots \otimes \hat \rho_N^{(n_s)}
\label{eq:shadow}
\end{equation}
where $\quad \hat \rho_i^{(n_s)}=3\left|s_i^{(n_s)}\right\rangle\left\langle s_i^{(n_s)}\right|-\mathbb{I}$ and $s_i^{(n_s)}$ is the outcome of qubit $i$ in the $n_s$-th randomized measurement. Eq.~\eqref{eq:shadow} allows in principle to fully recover the density matrix $\rho$, in the sense $\mathbb{E}[\hat \rho] = \rho$ where we take the expectation value of many random unitaries and projective measurements \cite{elben2022randomized}. However,  this requires an exponentially large number of copies of states $N_s=\mathcal{O}(\exp(N))$.  In contrast,  $N_s=\mathcal{O}(3^r \log (N) / \epsilon^2)$ is enough to provide an $\epsilon$-accurate approximation of all reduced $r$-body reduced density matrix, allowing to estimate expectation values of $r$-local observables. We emphasize that this estimation can be made robust against errors in the application of the random unitaries and read-out errors \cite{Chen2020,Koh2020,vitale2023estimation}.

\vspace{1em}
\noindent \textbf{Learning conservation laws.} Our algorithm makes use of classical shadow formalism to evaluate the expectation values of all geometrically $k$-local Pauli observables with randomized measurements (the measurement results can be used to estimate all $k$-local Pauli observables, but we only focus on the geometrically local ones), and then post-process the measurement data to identify the conserved quantities. This is a non-trivial task because there are uncountably many linear combinations of these Pauli observables that can possibly be conserved. Here we propose a method to efficiently narrow down the range of conserved quantities to look for.

We consider a quantum system on a $D$-dimensional lattice, with each site containing a qubit. We denote all geometrically $k$-local Pauli operators by $P_i$, $i=1,2,\cdots,N_P$ with
$N_P \leq \Or(N)$. We then look for conserved quantities of the form
\begin{equation}
\label{eq:conserved_quantity}
    O=\sum_{i}c_i P_i.
\end{equation}
In other words, $\{P_i\}$ form a basis of the subspace in which we search for conserved quantities. 

The expectation values of $P_i$ at time $t_j$, $j=1,2,\cdots,N_T$, form a data matrix of size $N_P\times N_T$ ($N_T=\mathcal{O}(N_P)$), which we denote by $X$. Its elements are
\begin{equation}
\label{eq:defn_X}
    X_{ij} = \braket{P_i(t_j)}.
\end{equation}

Our algorithm is built upon the following observation 
(see also Refs.~\cite{Bentsen2019integrable,li2020hamiltonian}): 
every conserved quantity of the form \eqref{eq:conserved_quantity} lies in the null space of a matrix $W^\top$, where its transpose matrix $W$ is defined through
\begin{equation}
\label{eq:defn_W}
    W_{ij} = \braket{P_i(t_j)} - \frac{1}{N_T}\sum_{j'} \braket{P_i(t_{j'})}.
\end{equation}
This is because, if operator $O$ defined in \eqref{eq:conserved_quantity} is conserved, then $\sum_{i}c_i \braket{P_i(t_j)}$ are equal for all $j$, and they are thus all equal to the average. Consequently
\begin{equation}
    \sum_i c_i \braket{P_i(t_j)} = \frac{1}{N_T}\sum_{ij'} c_i \braket{P_i(t_{j'})}.
\end{equation}
By \eqref{eq:defn_W} we then have $W^\top \vec{c}$=0, where $\vec{c}=(c_1,c_2,\cdots,c_{N_P})$ is the vector formed by the coefficients in $O$.

From the above analysis, we can see that all  conservation laws which are linear combinations of geometrically $k$-local terms must correspond to a vector in the null space of $W^\top$, and consequently, we can find all of them by examining this null space. 
At the same time, the dimension of this null space yields the number of independent conservation laws.
Each singular value $\sigma$ of $W$ describes how much its corresponding operator expectation value changes over time \cite{Bentsen2019integrable}. More precisely, let $u=(u_1,u_2,\cdots,u_{N_P})$ be the left singular vector corresponding to $\sigma$, and $O=\sum_i u_i P_i$, then
\begin{equation}
\label{eq:physical_meaning_of_sv}
    \sigma^2 = \sum_{j}\left|\braket{O(t_j)}-\frac{1}{N_T}\right.\sum_{j'} \left.\vphantom{\frac{1}{N_T}} \braket{O(t_{j'})}\right|^2,
\end{equation}
where $\overline{\braket{O(t_{j'})}}$ denotes the average of $\braket{O(t_{j'})}$ over the time index $j'$.

Because of the inevitable statistical noise, the matrix $W^{\top}$ we get from data will most likely not have a non-trivial null space. Therefore instead of looking at the null space, we will look at the subspace spanned by the left singular vectors of $W$ corresponding to singular values below a truncation threshold $\epsilon$, which serves as a precision parameter. These singular vectors are readily obtainable by performing SVD on our approximation of $W$ based on finitely many samples (see also Ref.~\cite{li2020hamiltonian} for a similar procedure in the context of Hamiltonian learning). The number of such singular values provides an upper bound of the number of independent conserved quantities, which we will prove later. The computational cost of performing SVD on the data matrix is $\mathcal{O}(N^3)$.

So far we have been mainly concerned with conserved quantities that are specific to a single initial state. We may also learn conserved quantities for a distribution $\mathcal{D}$ of initial states in a similar way. In this scenario, we not only sample times $t_j$, but also the initial states $\rho_k$ from $\mathcal{D}$ independently, for $k=1,2,\cdots,N_I$. The data matrix $X$ is constructed to have $N_P$ rows and $N_T N_I$ columns, consisting of entries $X_{i,jk}=\tr[P_i(t_j)\rho_k]$, where $j$ and $k$ together index the columns. The matrix $W$ is similarly modified to be $W_{i,jk} = X_{i,jk} - N_T^{-1} \sum_{j'}X_{i,j'k}$. 

\vspace{1em}
\noindent \textbf{Testing conservation laws.} The above procedure gives us candidates for conservation laws. However, it is not guaranteed that the quantities we get are indeed conserved, and therefore we need to test the candidates. Testing a finite group symmetry has been considered in Ref.~\cite{LabordeWilde2022quantum}, but their algorithm requires implementing the group action on a quantum computer, whereas we want to keep our procedure to only single-qubit operations. There are two problems that we need to overcome: the first is that in the above we only look at a discrete set of times $t_j$, and cannot rule out the possibility that some quantity be conserved at these discrete times but not conserved at other times. The second is that we cannot hope to tell if a quantity is exactly conserved because of the presence of statistical noise. Consequently, we formalize the problem into a hypothesis testing problem. We first define how far a quantity deviates from its average up to time $T$ by
\begin{equation}
    d(O,\rho) = \max_{t\in[0,T]}\Big| \tr[\rho O(t)] - \frac{1}{T}\int_0^T \tr[\rho O(s)]\dd s \Big|.
\end{equation}
With $d(O,\rho)$ we introduce the two hypotheses that we want to distinguish
\begin{equation}
\label{eq:hypotheses}
    \mathbb{E}_{\rho\sim\mathcal{D}} [ d(O,\rho) ]=0, \text{ or } \mathbb{E}_{\rho\sim\mathcal{D}} [ d(O,\rho) ]\geq \epsilon.
\end{equation}
for every candidate $O$ that comes from the learning procedure. The classical shadow technique enables us to process all $O$'s in parallel.

For a fixed $\rho\sim\mathcal{D}$, we compute the maximal deviation of observable $O$ from its time average using robust polynomial interpolation \cite{KaneEtAl2017robust}. We first randomly sample the discrete times $t_j$, and then perform robust polynomial interpolation to obtain values for $\braket{O(t)}$ at all times $t\in[0,T]$. Then we can directly compute the maximal deviation from the time average. This enables us to compute $d(O,\rho)$ with high confidence level.  Note that the above discussion is for continuous $t$. If we want to test conservation laws for discrete $t$, such as for a Floquet system, then the problem comes strictly easier, as interpolation will not be needed. The ensemble average $\mathbb{E}_{\rho\sim\mathcal{D}}d(O,\rho)$ can be computed from finitely many samples of $\rho$, thus enabling us to solve the hypothesis testing problem in Eq.~\eqref{eq:hypotheses}.

\section{Rigorous Guarantees}
\label{section3}

As noted previously, the estimates for $\tr[P_i(t_j)\rho_k]$ necessarily involves statistical noise. We will then analyze how the noise impacts the result we get. First we will analyze that in the learning algorithm based on finding the null space, the algorithm still yields an upper bound of the number of conserved quantities even when statistical noise is present.

We denote the number of independent conserved quantities by $N_c$, and the dimension of the null space of $W^{\top}$ by $D_{\mathrm{null}}$. It is guaranteed that $N_c\leq D_{\mathrm{null}}$ because a quantity that is conserved in all times must also be conserved at discrete times $t_j$ and for the sampled states $\rho_k$. 
$W$ is computed from the data matrix $X$, but in practice, we do not directly have access to $X$, but can only obtain its noisy estimate $\hat{X}$, which leads to a noisy estimate for $W$ that we denote by $\hat{W}$. $\hat{W}$ is almost surely full-rank due to the effect of the noise. 
We denote $E=\hat{X}-X$, and this is the matrix containing all the entry-wise errors. The matrix $W$ is perturbed similarly, and the errors can be collected into a matrix whose spectral norm is at most $\|E\|$. 
Consequently, we cannot directly estimate $D_{\mathrm{null}}$. As discussed before, instead we look at the number of singular values of $\hat{W}$ that are below a threshold $\epsilon$, which we denote by $\hat{D}_{\mathrm{null}}$. For $\hat{D}_{\mathrm{null}}$ we have the following theorem (where for simplicity, we let $N_T\times N_I$ and $N_P$ be of order $\Or(N)$, more general $N_T$, $N_I$ and $N_P$ are considered in  Theorems 3 and 4 of Sec.~I of the supplemental material (SM) \cite{SM}:

\begin{thm}
\label{thm:counting_conservation_laws}
With $\wt{\Or}(N^3 \epsilon^{-2}\log(\delta^{-1}))$ samples~\footnote{We use the asymptotic notation $\tilde{\Or}(f(x))$ to denote $\Or(f(x)\mathrm{polylog}(f(x)))$.}, we can compute an integer $\hat{D}_{\mathrm{null}}^{\mathrm{median}}$ satisfying $
    N_c\leq \hat{D}_{\mathrm{null}}^{\mathrm{median}}
$, where $N_c$ is the number of conserved quantities, 
with probability at least $1-\delta$. Here $\hat{D}_{\mathrm{null}}^{\mathrm{median}}$ is the median taken over $\Or(\log(\delta^{-1}))$ independent samples of $\hat{D}_{\mathrm{null}}$ and $\hat{D}_{\mathrm{null}}$  denotes the number of singular values of $\hat{W}$ below $\epsilon$. In particular, when the quantum system has constant correlation length, the sample complexity can be reduced to  $\wt{\Or}(N^2 \epsilon^{-2}\log(\delta^{-1}))$.
\end{thm}
For the proof of this theorem, we refer to Sec.~I of the SM \cite{SM}, in which we  bound singular value perturbation that comes from statistical noise.
From its definition, we can see $\hat{D}_{\mathrm{null}}$ that is a decreasing function of $\epsilon$, and consequently, for smaller $\epsilon$, we will have a tighter upper bound for $D_{\mathrm{null}}$ and the number of conserved quantities $N_c$. 
If we keep the singular value perturbation below $\epsilon$, then the $D_{\mathrm{null}}$ 0-singular values of $W$ will still be below $\epsilon$ after perturbation, thus ensuring $D_{\mathrm{null}}\leq \hat{D}_{\mathrm{null}}^{\mathrm{median}}$. This will require more samples as $\epsilon$ decreases,  as can be seen from Theorem~\ref{thm:counting_conservation_laws}.

We can also guarantee that by collecting all the left-singular vectors of $\hat{W}$ corresponding to singular values below the threshold $\epsilon$, we will have all the conservation laws approximately contained in the span. More precisely, each conservation law, when expressed as a norm-1 vector, will have an overlap with the subspace spanned by these singular vectors, and this overlap is lower bounded by $\sqrt{1-\|E\|^2/\epsilon^2}$.
Therefore, when $\|E\|\ll\epsilon$, we will have an accurate description of all conservation laws. For detailed proof of this bound, see Sec.~II of the SM \cite{SM}.

Next, we will provide guarantees that the candidates for conserved quantities from the learning algorithm can be efficiently verified using the procedure described previously. First, we consider the scenario where the conservation law is specific to a single initial state, i.e., the distribution $\mathcal{D}$ is completely concentrated on $\rho$.
\begin{thm}
    \label{thm:testing_conservation_laws}
    We assume that $\mathcal{D}$ is concentrated on a single $\rho$.
    Let $f_i(t)=\Tr[\rho(t)O_i]$, for $i=1,2,\cdots,\chi$. We further assume that $\left|\frac{\dd^\ell f_i(t)}{\dd t^\ell}\right|\leq \Or(\Gamma^\ell \ell!)$ for all $\ell\geq 1$. Then for $T>0$ we can distinguish between the two hypotheses in \eqref{eq:hypotheses} for each $i$ with probability at least $1-\delta$ using 
    $
        \wt{\Or}\left(\Gamma T \epsilon^{-2}\log(\delta^{-1})\max_i\|O_i\|_{\mathrm{shadow}}^2\right)
    $
    samples.
\end{thm}
We refer to  Sec.~III of the SM \cite{SM} for detailed proof.
We note that $\left|\frac{\dd^\ell f_i(t)}{\dd t^\ell}\right|\leq \Or(\Gamma^\ell \ell!)$ is a very reasonable assumption to make. We will show in Sec.~VI
of the SM \cite{SM} that this assumption holds with $\Gamma=\Or(1)$ 
when the dynamics is described by the von Neumann equation, and the Hamiltonian satisfies certain conditions.
These Hamiltonians include geometrically local Hamiltonians and certain power-law interaction Hamiltonians. Without such an assumption, we can also choose $\Gamma=\|H\|$ and then this inequality holds for all Hamiltonians. An extension to the Lindbladian case is straightforward. 

Next, we consider a generic initial state distribution $\mathcal{D}$. In this scenario, we can sample $\rho_k$, $k=1,2,\cdots,N_I$ from the distribution $\mathcal{D}$, and test if the observables $O_1, O_2, \cdots, O_{\chi}$ are conserved for the sampled initial states. This naturally leads to the question of whether we can generalize the testing results for the sampled initial states to the entire distribution. The above involves generalization errors of the form
\begin{equation}
    \Big|\mathbb{E}_{\rho\sim\mathcal{D}}d(O_i,\rho)-\frac{1}{N_I}\sum_{k=1}^{N_I}d(O_i,\rho_k)\Big|.
\end{equation}
In Sec.~IV in the SM \cite{SM}, we show that we can use $N_I=\Or(\epsilon^{-2}\log(\chi \delta^{-1})\max_i\|O_i\|^2)$ to ensure that the generalization errors for all observables are below $\epsilon/4$ with probability at least $1-\delta/2$. For each sampled $\rho_k$, we need $\wt{\Or}\left(\Gamma T \epsilon^{-2}\log(\delta^{-1})\max_i\|O_i\|_{\mathrm{shadow}}^2\right)$ samples to compute $d(O_i,\rho_k)$, $i=1,2,\cdots,\chi$ according to Theorem~\ref{thm:testing_conservation_laws}, which multiplied by $N_I$ yields the total sample complexity for estimating $\mathbb{E}_{\rho\sim\mathcal{D}}d(O_i,\rho)$ up to precision $\epsilon/2$. Therefore
\begin{thm}
    \label{thm:test_conservation_laws_multiple_initial_states}
    Under the same assumptions as in Theorem~\ref{thm:testing_conservation_laws}, except that we do not restrict the form of the initial state distribution $\mathcal{D}$, the hypothesis testing problem in \eqref{eq:hypotheses} can be solved using
    $
    \wt{\Or}\left(\Gamma T \epsilon^{-4}\log(\delta^{-1})\log(\chi \delta^{-1})\max_i\|O_i\|^2\max_i\|O_i\|_{\mathrm{shadow}}^2\right)
    $ 
    samples. 
\end{thm}

So far, we have considered quantities that are on average conserved for an ensemble of states. It is natural to ask whether we can determine if an observable $O$ is conserved for all states, namely $[H,O]=0$. Through a quantum query complexity lower bound, we can show that this task cannot be accomplished efficiently in the worst case. The high-level idea of this argument goes as follows: suppose we have a black-box oracle $U$ encoding a bit-string $\boldsymbol{x}=(\boldsymbol{x}_1,\boldsymbol{x}_2,\dots)$ through $U\ket{n} = (-1)^{\boldsymbol{x}_n}\ket{n}$, then letting $H=U$ we can implement $e^{-iHt}$ using two queries to $U$. If an algorithm can distinguish between $\|[H,O]\|=0$ or $\geq 1$ with high probability with $Q$ queries to $e^{-iHt}$, we can then show that it can evaluate $\operatorname{OR}(\boldsymbol{x})$ with $2Q$ queries to $U$. The query complexity lower bound of the OR function \cite{beals2001quantum} then tells us that $Q=\Omega(2^{N/2})$. For a detailed statement of the result and its proof, see Sec.~V of the SM \cite{SM}.

\section{Numerical experiments}

In this section, we illustrate our algorithm with numerical examples. We consider a $\mathbb{Z}_2$ gauge theory and a disordered Heisenberg model in one dimension.

\subsection{Identifying conservation laws in a lattice gauge theory}

\begin{figure}
    \centering
    \includegraphics[width=\linewidth]{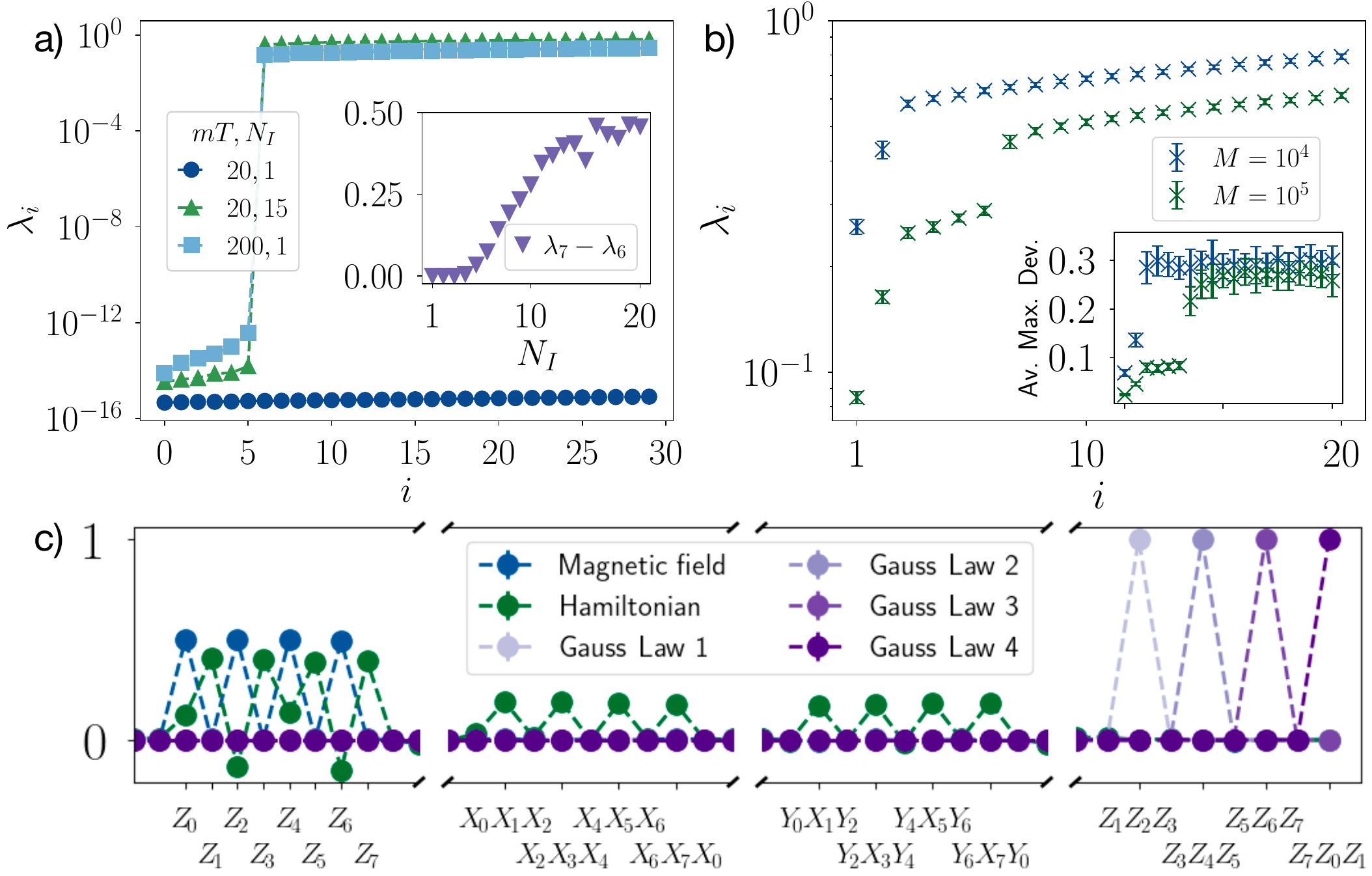}
    \caption{\textit{Learning conservations laws in a $\mathbb{Z}_2$-Gauge theory.} In panel a), we display the 30 smallest eigenvalues of the exact data matrix  for various choices of evolution times $T$ and number of initial states $N_I$. At sufficiently long times $T$, $N+2$ singular values are gapped out (light and dark blue). The use of multiple initial states (green triangles)  allows to shorten the evolution time considerably.  Inset illustrates this effect, showing the gap as a function of the number of initial states for fixed time $T=20$.
    In panel b), we display the 20 smallest eigenvalues of the noisy data matrix, construct from $M=10^4$ and $M=10^5$ per initial state ($N_I=15$) and final time $T=20, N_T=41$. For $M=10^4$ ($M=10^5$), increasing $M$,  two (six) singular are gapped out. Testing these against independently obtained data yields small variations over time (inset).  In panel c), we display the Pauli basis expansion of the corresponding lowest 6 singular vectors. We identify the magnetization, Hamiltonian, and four Gauss laws.
    In all panels, $N=8$ and $N_T = 2 N_P / N_I  = 624$. Points with error bars are the average and standard deviation over 25 experiments with identical parameters. 
    }
    \label{fig:1}
\end{figure}

As a first numerical example, we consider a $\mathbb{Z}_2$ lattice gauge theory with staggered matter fields in one spatial dimension. The Hamiltonian is specified by
\begin{align}
    H_{\mathbb{Z}_2}=&\frac{1}{2a} \sum_{i=0}^{N/2-1}  \left(\sigma_{2i}^+\sigma_{2i+1}^x \sigma_{2i+2}^- + \text{h.c.} \right) \nonumber \\ &+ m \sum_{i=0}^{N/2-1} \frac{(-1)^i}{2} \left( \mathbb{I}_2 + \sigma^z_{2i} \right) + e \sum_{i=0}^{N/2-1} \sigma^z_{2i+1}.
\end{align}
where we choose periodic boundary conditions. By direct inspection of the Hamiltonian, we find that we can expect $N/2+2$ conservations laws, given by the Hamiltonian itself, the magnetization
\begin{align}
M = \frac{1}{2} \sum_{i=0}^{N/2-1} \left( \mathbb{I}_2 + \sigma^z_{2i} \right) 
\end{align}
and $N/2$ Gauss laws
\begin{align}
    G_{2j} = e^{i\pi Q_{2j} } \sigma^z_{2j-1}\sigma^z_{2j+1}.
\end{align}
We note that magnetization and Hamiltonian are linear combinations of geometrically $1$-local and $3$-local terms, which  have support on the entire system. In contrast, the $N/2$ Gauss laws are strictly geometrically $3$-local. In the following, we consider $N=8$ qubits,  set the mass parameter to unity  $m=1$, and choose electric field $e=3/2m$ and lattice spacing  $a=m/3$. 

First, we investigate our protocol in the absence of statistical noise due to a finite number of measurements [Fig.~1a)]. We aim to learn all conservation laws with weight $k\leq 3$. We find that for sufficiently long times $T=200$, data collected from dynamics starting from a single initial random product state $N_I=1$ is sufficient to identify all expected $N+2$  conservation laws: The singular values $\lambda_{i}$ for $1\leq i \leq N +2$  are close to zero (a non-zero value originates from finite machine precision) with a large gap to $\lambda_{N+3}$. At shorter times $T=20$ and $N_I=1$, the spectrum of singular values appears to be continuous, and conservation laws are not apparent.
In contrast, data collected from several random initial product states is substantially more expressible (see Ref.~\cite{evans2019scalable} for a similar observation in the context of Hamiltonian learning). Even at short times $T=20$, the expected conservation laws can be identified. We emphasize that hereby, we keep the total number of points $N_T N_I \approx 2N_P$ where data is taken to be constant to enable a fair comparison ($N_T=2 N_P = 624$ for $N_I=1$ and $N_T =41$ for $N_I=15$). In all cases, the time points are equally spaced. 

Secondly, we simulate our complete protocol, including a finite number of measurements $M$ per time point and initial state. We choose $N_I = 15$ randomly chosen initial product states, a total time evolution time of $T=20/m$ with $N_T=41$ steps. We find that for a moderate number of  $M= 10^5$ randomized measurements, $N+2$ singular values are gapped out [Fig.~1b)]. Analyzing a Pauli basis expansion of the  corresponding singular vectors of our data matrix, we find that these correspond to the expected conservation laws, magnetization, energy (Hamiltonian) and $N$ Gauss laws [Fig.~1c)]. 

Finally, to test that the learned quantities are indeed conserved over times, we employ an independent data set of same size. We estimate expectation values of the learned quantities at different times and compute their maximum deviation from their mean values, averaged over initial states. We note that this represents a simplified testing procedure than employed in Secs.~\ref{section2} and \ref{section3}, and devote the full numerical implementation of the robust polynomial interpolation for testing to future work. Indeed, we find that quantities corresponding to small singular values have small variations over time up to statistical noise originating from a finite number of randomized measurements $M$. 
The testing procedure also helps us better distinguish conservation laws from observables that are close to being conserved, by opening up the gap between singular values, as can be seen in Figure~\ref{fig:1}b. By using a different set of data to perform testing, we can exclude quantities with only small variation for a single noise realization or a single set of initial states. This is similar to detecting overfitting in supervised learning.

While we have so far concentrated on unitary dynamics, we emphasize that our protocol can serve to learn conservation laws of arbitrary quantum dynamics. To illustrate this point, we add local dephasing with strength $\gamma$, corresponding to jump operators $L_i=\sqrt{\gamma}\sigma^z_i$, and solve the corresponding Lindblad equation (all other Hamiltonian parameters remain the same). Since magnetization and Gauss laws are diagonal in the computational $Z$-basis, they remain conserved also for $\gamma>0$. In contrast, energy conservation is lost as shown in Fig.~2a where $\lambda_2$ is missing (the indices of the subsequent singular values $\lambda_3, \lambda_4, \dots$ have been shifted by $1$ for clarity). 

Our numerical experiments demonstrate that we can learn conservation laws, which are linear combinations of $k$-local Pauli strings with a moderate number of randomized measurements $M$. We can decrease  the required number of measurements further if we restrict ourselves to learning conservation laws with support on subsystems only, i.e.\ disregard quantities such as magnetization or Hamiltonian, which are linear combinations of few body terms but have support on the entire system. To achieve this, we construct reduced data matrices $W_A$ from measurement data obtained from the subsystem  $A$ only. This is illustrated in Fig.~2b, where we plot the singular values of data matrices $W_{A_j}$ with $A_j = [j-1,j,j+1]$ containing the three sites $j-1,j,j+1$ as function of j (periodic boundary conditions are implied). With only $M=10^4$ randomized measurements [c.f.~$M=10^5$ in Fig.~1b)]
per time point and initial state, we can identify the expected $N$ Gauss laws contained in subsystems $A_j$ with $j \mod 2 =0$.

\begin{figure}
    \centering
    \includegraphics[width=\linewidth]{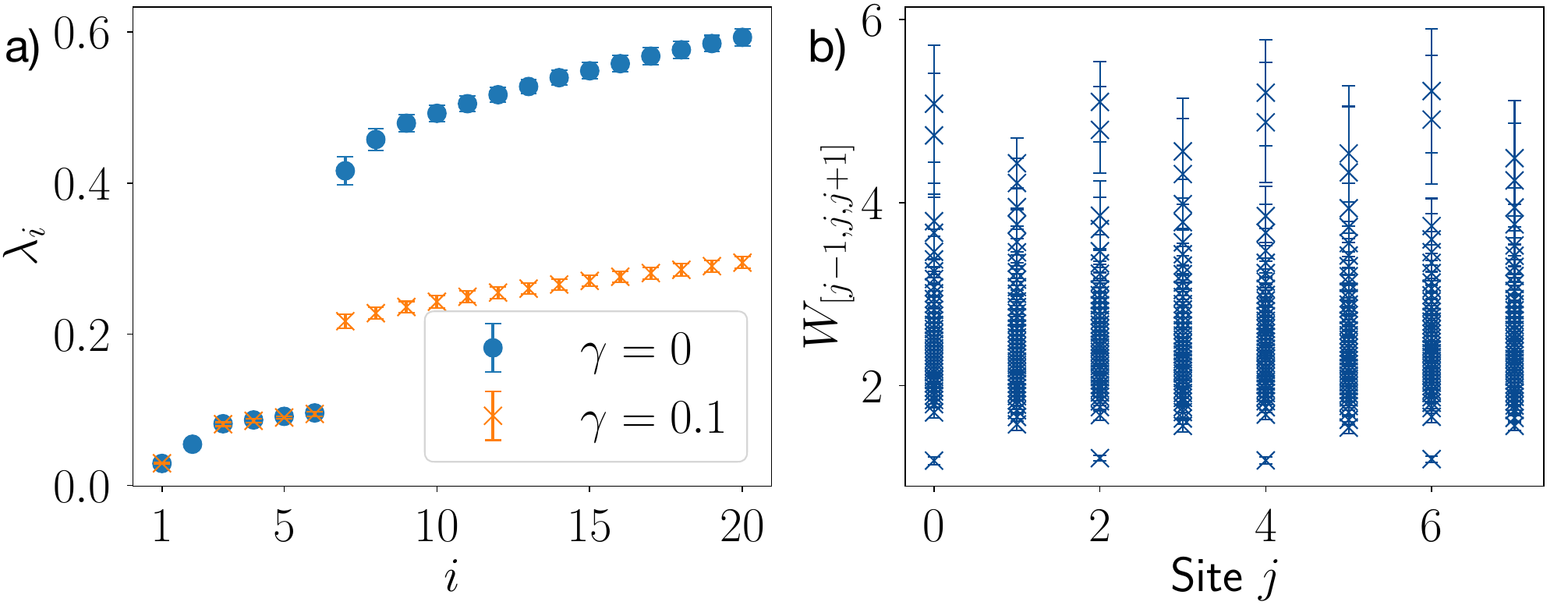}
    \caption{\textit{Learning conservation laws in open quantum systems.} In panel a), we display the 20 smallest singular values of $W$ for both Hamiltonian (blue) and Lindbladian dynamics (orange) with local dephasing with rate $\gamma=0.1$. While magnetization and Gauss laws are conserved in both cases, the Hamiltonian (energy) is only conserved for Hamiltonian dynamics ($\lambda_2$ is absent for $\gamma =0.1$).  We choose $M=10^6$ measurements and use a Gaussian noise approximation to simulate the resulting shot noise. In all panels, $N=8$, $N_T  = 41$, $N_I=15$ and points with errorbars are the average and standard deviation over 25 experiments with identical parameters. Panel b) displays the singular values obtained from the data matrix $W_{[j-1,j,j-1]}$ restricted to three adjacent subsystem sites as function of $j$. This allows to learn the Gauss laws, corresponding to the gapped singular values at even $j$, with considerably fewer measurements $M=10^4$ (c.f.\ Fig.~\ref{fig:1}b).} 
    \label{fig:2}
\end{figure}

\subsection{Identifying conservation laws in a many-body-localized system}

Next, we consider a disordered one-dimensional XXZ-spin chain with nearest neighbor interactions which serves as  a standard model for investigating many-body localization and thermalization \cite{https://doi.org/10.1002/andp.201700169,doi:10.1146/annurev-conmatphys-031214-014726,RevModPhys.91.021001}. It is described by the following  Hamiltonian
\begin{align}
    &H_{\text{XXZ}}\nonumber \\
    &=J_x \sum_i\left(\sigma_i^x \sigma_{i+1}^x+\sigma_i^y \sigma_{i+1}^y\right)+J_z \sum_i \sigma_i^z \sigma_{i+1}^z+\sum_i h_i \sigma_i^z
    \label{eq:XXZ}
\end{align}
where the local disorder potentials  $h_i$ ($i=1,\dots N$) are randomly distributed in the interval $[-w,w]$, and $w$ is the disorder strength. We assume in the following $J_x=J_z=1$.

\begin{figure}
    \centering
    \includegraphics[width=0.38\textwidth]{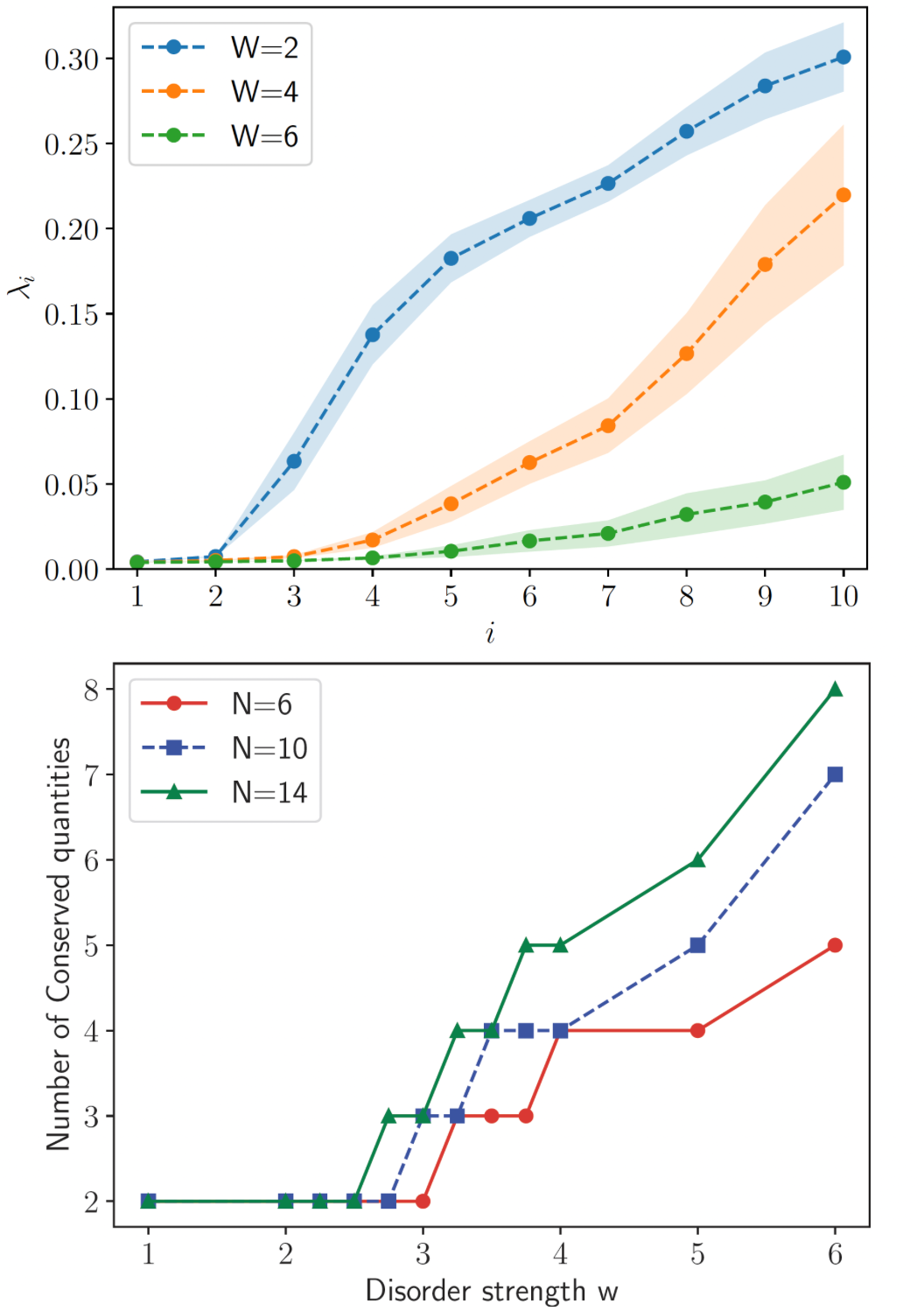}
    \caption{\textit{Learning conservation laws in disordered spin-chains.}(a) We display the 10 lowest squared singular values of the data matrix $\hat{W}$ constructed using $N_I=17$ initial states evolved according to $H_{\text{XXZ}}$ with  different disorder strengths $w=2,4,6$ (blue, orange, green)  to a final time $T=40$ and evaluated at $N_T=41$ equally spaced timepoints. We add Gaussian noise, simulating shot noise arising from $M=5 \times 10^5$ measurements. Each point is averaged over 11 random disorder patterns, the error bars indicate the standard error of the mean. The two smallest singular values $\lambda_1,\lambda_2$ correspond to magnetization and Hamiltonian, respectively. As the disorder strength $w$ becomes larger, the singular values $\lambda_{i\geq 3}$ decrease significantly indicating an increasing number of approximate conservation laws. (b) We show the number of singular values below a threshold $\epsilon=0.02$ as function of $w$ for different system sizes. As in panel(a), we choose $M=5\times10^5, T=40, N_T=41$. The number of initial states for system size $N=6,10,14$ is $N_I=9,17,24$, respectively, to ensure $N_I N_T\approx2N_P$. Points correspond to the median of 11 random Hamiltonian configurations. We observe a sharp increase with increasing disorder strength.}
    \label{fig:3}
\end{figure}

Numerical studies \cite{PhysRevB.82.174411,znidaric2008many,PhysRevLett.111.127201,luitz2015many} indicate that this model exhibits a transition from an ergodic phase at weak disorder to an ergodicity breaking many-body localized phase at sufficiently strong disorder \footnote{The precise transition disorder strength in the thermodynamic limit is subject to ongoing research, see e.g.~\cite{long2023phenomenology} and references therein.}. While in the ergodic phase, nearly all eigenstates obey the Eigenstate Thermalization Hypothesis \cite{deutsch1991quantum,srednicki1994chaos}  in the MBL phase ETH is not valid, and the system is characterized by an extensive number of \emph{quasi-local} conservation laws $\tau_i$  ($i=1,\dots N$), called l-bits \cite{PhysRevLett.111.127201,PhysRevB.90.174202,ROS2015420,imbrie2017local}. In terms of these conservation laws, the Hamiltonian \eqref{eq:XXZ} can be written as 
\begin{equation}
H_{\text{diag}}=\sum_i \xi_i {\tau}_i+\sum_{i<j} J_{i j} {\tau}_i {\tau}_j+\sum_{i<j<k} J_{i k j } {\tau}_i {\tau}_j {\tau}_k+\ldots
\label{eq:lbits}
\end{equation}
where, in the MBL phase, the ${\tau}_i$ are quasi-local, in the sense that they can be approximated by geometrically local operators to exponential precision, and the coupling coefficients   $J_{ij}$, $J_{ikj}$ decay exponentially with distance $|i-j|$. We emphasize that, by switching to an energy eigenbasis, we can always rewrite the Hamiltonian \eqref{eq:XXZ} in the form of Eq.\ \eqref{eq:lbits}, also in the thermalizing phase. In general, however, the conservation laws  ${\tau}_i$ will be completely non-local, high-weight operators  which vanishing overlap to the microscopic degrees of freedom $\sigma^z_i$.  In this case, Eq.~\eqref{eq:lbits} is of little use. 
Finally, we note that independent of the disorder strength, there are always two conserved quantities given as sums of local observables, which are the Hamiltonian itself and the total magnetization.

With our learning algorithm, we target conserved quantities given as sums of local observables. While we expect to be able to learn the local XXZ-Hamiltonian and total magnetization for all disorder strengths, the conserved quantities  ${\tau}_i$ are  expected to be inaccessible in the thermal phase due to their non-local nature. In contrast, for strong disorder, we expect an extensive amount of  approximately conserved local quantities, approximating the quasi-local l-bits ${\tau}_i$ to high precision. To test these expectations in small systems, we simulate our protocol, by sampling the random initial product state, picking a random disorder pattern and time-evolve under the corresponding XXZ-Hamiltonian Eq.~\eqref{eq:XXZ} using exact diagonalization. We evaluate Pauli expectations of all $N_P$ geometrically up to $3$-local Pauli operators at $N_T=41$  equidistant time points up to a final time $T=40$. We repeat this for $N_I\approx2N_P/N_T$ initial states per fixed disorder pattern.  To simulate realistic shot noise arising from a finite number of randomized measurements, we use a Gaussian noise approximation adding independent Gaussian noise to each expectation value with a variance corresponding to $M=5\times10^5$ randomized measurements per time point and initial state \footnote{In practice, we reconstruct all Pauli expectation values at a given time from the same experimental randomized measurement data. Thus, shot noise on different expectation values is in principle correlated. For large number of measurements, we however expect these correlation to be weak, and indeed confirm numerically that results obtained using the Gaussian noise approximation and actual randomized measurements are consistent.}. Finally, we construct the data matrix $\hat{W}$ and perform SVD. 

In Fig.~\ref{fig:3}a, we display the  10 lowest singular values of $\hat{W}$ as a function of disorder strength, each point corresponds to an average over 11 random Hamiltonian configurations. Consistent, with our expectation, we find independent of the disorder strength two small singular values, corresponding to magnetization and Hamiltonian. As remarked before, they still attain a non-zero value due to the finite number of measurements $M$. In addition, we find that the singular values $\lambda_i$ ($i\geq 3$) decrease strongly with increasing disorder strength, indicating an increasing number of approximately conserved quantities. To show this more quantitatively, we plot in Fig.~\ref{fig:3}b), the number of singular values  below a threshold $\epsilon=0.02$ as a function of the disorder strength for various system sizes. We observe a sharp increase at a disorder strength, which is consistent with previous findings on the onset of many-body localization effects in finite-size systems \cite{PhysRevB.82.174411,znidaric2008many,PhysRevLett.111.127201,luitz2015many}. In addition, the number of singular values below the threshold increases with system size, indicating indeed an extensive number of approximately conserved quantities. 

\section{Conclusion and Outlook}
\label{section9}
In this paper, we propose a method for learning conservation laws in arbitrary quantum dynamics. 
Our method can find a set of observables that include all conservation quantities with high probability, and we also propose a method to test the conservation law candidates obtained in this way. The sample complexity and classical processing time are both at most polynomial in the system size.
The conservation laws hold for either a single input state or an ensemble of input states, and for the latter case, we derive a generalization bound ensuring that the result from finitely many samples can be reliably generalized to the entire ensemble.
We provided a proof of principle of our method using numerical experiments in a one-dimensional $\ZZ_2$ lattice gauge theory and one-dimensional MBL systems.
Beyond these examples, we envision a wide range of applications for our protocol, ranging from Hilbert space fragmentation \cite{moudgalya2022quantum}  to random circuits with symmetries \cite{fisher2023random} and the study of general quantum channels \cite{albert2019asymptotics}. In addition, knowledge of conserved quantities in dynamics enables powerful error mitigation techniques for NISQ devices \cite{cai2023quantum} and more efficient (randomized measurement) protocols  for probing many-body entanglement \cite{elben2018renyi,bringewatt2023randomized}.

\begin{acknowledgments}
The authors thank the inspiring discussions with John Preskill, Elmer V.H.\ Doggen and Laimei Nie and Benoit Vermersch. Y.Z. acknowledges funding from the National Science Foundation (PHY-1733907). The Institute for Quantum Information and Matter is an NSF Physics Frontiers Center. Y.T.\ acknowledges funding from the U.S. Department of Energy Office of Science, Office of Advanced Scientific Computing Research, (DE-NA0003525, and DE-SC0020290). Work supported by DE-SC0020290 is supported by the DOE QuantISED program through the theory consortium ``Intersections of QIS and Theoretical Particle Physics'' at Fermilab. The work of A.E. was performed in part at the Aspen Center for Physics, which is supported by the National Science Foundation grant PHY-2210452.
Furthermore, A.E.\ acknowledges funding by the German National Academy of Sciences Leopoldina under the grant number LPDS 2021-02 and by the Walter Burke Institute for Theoretical Physics at Caltech.
\end{acknowledgments}

\bibliography{apssamp}

\providecommand{\noopsort}[1]{}\providecommand{\singleletter}[1]{#1}%
\begin{thebibliography}{79}%
\makeatletter
\providecommand \@ifxundefined [1]{%
 \@ifx{#1\undefined}
}%
\providecommand \@ifnum [1]{%
 \ifnum #1\expandafter \@firstoftwo
 \else \expandafter \@secondoftwo
 \fi
}%
\providecommand \@ifx [1]{%
 \ifx #1\expandafter \@firstoftwo
 \else \expandafter \@secondoftwo
 \fi
}%
\providecommand \natexlab [1]{#1}%
\providecommand \enquote  [1]{``#1''}%
\providecommand \bibnamefont  [1]{#1}%
\providecommand \bibfnamefont [1]{#1}%
\providecommand \citenamefont [1]{#1}%
\providecommand \href@noop [0]{\@secondoftwo}%
\providecommand \href [0]{\begingroup \@sanitize@url \@href}%
\providecommand \@href[1]{\@@startlink{#1}\@@href}%
\providecommand \@@href[1]{\endgroup#1\@@endlink}%
\providecommand \@sanitize@url [0]{\catcode `\\12\catcode `\$12\catcode
  `\&12\catcode `\#12\catcode `\^12\catcode `\_12\catcode `\%12\relax}%
\providecommand \@@startlink[1]{}%
\providecommand \@@endlink[0]{}%
\providecommand \url  [0]{\begingroup\@sanitize@url \@url }%
\providecommand \@url [1]{\endgroup\@href {#1}{\urlprefix }}%
\providecommand \urlprefix  [0]{URL }%
\providecommand \Eprint [0]{\href }%
\providecommand \doibase [0]{https://doi.org/}%
\providecommand \selectlanguage [0]{\@gobble}%
\providecommand \bibinfo  [0]{\@secondoftwo}%
\providecommand \bibfield  [0]{\@secondoftwo}%
\providecommand \translation [1]{[#1]}%
\providecommand \BibitemOpen [0]{}%
\providecommand \bibitemStop [0]{}%
\providecommand \bibitemNoStop [0]{.\EOS\space}%
\providecommand \EOS [0]{\spacefactor3000\relax}%
\providecommand \BibitemShut  [1]{\csname bibitem#1\endcsname}%
\let\auto@bib@innerbib\@empty
\bibitem [{\citenamefont {Carleo}\ \emph {et~al.}(2019)\citenamefont {Carleo},
  \citenamefont {Cirac}, \citenamefont {Cranmer}, \citenamefont {Daudet},
  \citenamefont {Schuld}, \citenamefont {Tishby}, \citenamefont
  {Vogt-Maranto},\ and\ \citenamefont {Zdeborov{\'{a} }}}]{carleo2019machine}%
  \BibitemOpen
  \bibfield  {author} {\bibinfo {author} {\bibfnamefont {G.}~\bibnamefont
  {Carleo}}, \bibinfo {author} {\bibfnamefont {I.}~\bibnamefont {Cirac}},
  \bibinfo {author} {\bibfnamefont {K.}~\bibnamefont {Cranmer}}, \bibinfo
  {author} {\bibfnamefont {L.}~\bibnamefont {Daudet}}, \bibinfo {author}
  {\bibfnamefont {M.}~\bibnamefont {Schuld}}, \bibinfo {author} {\bibfnamefont
  {N.}~\bibnamefont {Tishby}}, \bibinfo {author} {\bibfnamefont
  {L.}~\bibnamefont {Vogt-Maranto}},\ and\ \bibinfo {author} {\bibfnamefont
  {L.}~\bibnamefont {Zdeborov{\'{a} }}},\ }\bibfield  {title} {\bibinfo {title}
  {Machine learning and the physical sciences},\ }\href
  {https://doi.org/10.1103%2Frevmodphys.91.045002} {\bibfield  {journal}
  {\bibinfo  {journal} {Rev. Mod. Phys.}\ }\textbf {\bibinfo {volume} {91}}
  (\bibinfo {year} {2019})}\BibitemShut {NoStop}%
\bibitem [{\citenamefont {Udrescu}\ and\ \citenamefont
  {Tegmark}(2020)}]{doi:10.1126/sciadv.aay2631}%
  \BibitemOpen
  \bibfield  {author} {\bibinfo {author} {\bibfnamefont {S.-M.}\ \bibnamefont
  {Udrescu}}\ and\ \bibinfo {author} {\bibfnamefont {M.}~\bibnamefont
  {Tegmark}},\ }\bibfield  {title} {\bibinfo {title} {Ai feynman: A
  physics-inspired method for symbolic regression},\ }\href
  {https://doi.org/10.1126/sciadv.aay2631} {\bibfield  {journal} {\bibinfo
  {journal} {Sci. Adv.}\ }\textbf {\bibinfo {volume} {6}},\ \bibinfo {pages}
  {eaay2631} (\bibinfo {year} {2020})}\BibitemShut {NoStop}%
\bibitem [{\citenamefont {de~Silva}\ \emph {et~al.}(2020)\citenamefont
  {de~Silva}, \citenamefont {Higdon}, \citenamefont {Brunton},\ and\
  \citenamefont {Kutz}}]{10.3389/frai.2020.00025}%
  \BibitemOpen
  \bibfield  {author} {\bibinfo {author} {\bibfnamefont {B.~M.}\ \bibnamefont
  {de~Silva}}, \bibinfo {author} {\bibfnamefont {D.~M.}\ \bibnamefont
  {Higdon}}, \bibinfo {author} {\bibfnamefont {S.~L.}\ \bibnamefont
  {Brunton}},\ and\ \bibinfo {author} {\bibfnamefont {J.~N.}\ \bibnamefont
  {Kutz}},\ }\bibfield  {title} {\bibinfo {title} {Discovery of physics from
  data: Universal laws and discrepancies},\ }\href
  {https://www.frontiersin.org/articles/10.3389/frai.2020.00025} {\bibfield
  {journal} {\bibinfo  {journal} {Front. Artif. Intell.}\ }\textbf {\bibinfo
  {volume} {3}} (\bibinfo {year} {2020})}\BibitemShut {NoStop}%
\bibitem [{\citenamefont {Liu}\ \emph {et~al.}(2022)\citenamefont {Liu},
  \citenamefont {Madhavan},\ and\ \citenamefont
  {Tegmark}}]{PhysRevE.106.045307}%
  \BibitemOpen
  \bibfield  {author} {\bibinfo {author} {\bibfnamefont {Z.}~\bibnamefont
  {Liu}}, \bibinfo {author} {\bibfnamefont {V.}~\bibnamefont {Madhavan}},\ and\
  \bibinfo {author} {\bibfnamefont {M.}~\bibnamefont {Tegmark}},\ }\bibfield
  {title} {\bibinfo {title} {Machine learning conservation laws from
  differential equations},\ }\href
  {https://doi.org/10.1103/PhysRevE.106.045307} {\bibfield  {journal} {\bibinfo
   {journal} {Phys. Rev. E}\ }\textbf {\bibinfo {volume} {106}},\ \bibinfo
  {pages} {045307} (\bibinfo {year} {2022})}\BibitemShut {NoStop}%
\bibitem [{\citenamefont {Lu}\ \emph {et~al.}(2023)\citenamefont {Lu},
  \citenamefont {Dangovski},\ and\ \citenamefont
  {Solja{\v{c}}i{\'{c}}}}]{lu2022discovering}%
  \BibitemOpen
  \bibfield  {author} {\bibinfo {author} {\bibfnamefont {P.~Y.}\ \bibnamefont
  {Lu}}, \bibinfo {author} {\bibfnamefont {R.}~\bibnamefont {Dangovski}},\ and\
  \bibinfo {author} {\bibfnamefont {M.}~\bibnamefont {Solja{\v{c}}i{\'{c}}}},\
  }\bibfield  {title} {\bibinfo {title} {Discovering conservation laws using
  optimal transport and manifold learning},\ }\href
  {https://doi.org/10.1038%2Fs41467-023-40325-7} {\bibfield  {journal}
  {\bibinfo  {journal} {Nature Communications}\ }\textbf {\bibinfo {volume}
  {14}} (\bibinfo {year} {2023})}\BibitemShut {NoStop}%
\bibitem [{\citenamefont {King}\ \emph {et~al.}(2009)\citenamefont {King},
  \citenamefont {Rowland}, \citenamefont {Oliver}, \citenamefont {Young},
  \citenamefont {Aubrey}, \citenamefont {Byrne}, \citenamefont {Liakata},
  \citenamefont {Markham}, \citenamefont {Pir}, \citenamefont {Soldatova},
  \citenamefont {Sparkes}, \citenamefont {Whelan},\ and\ \citenamefont
  {Clare}}]{doi:10.1126/science.1165620}%
  \BibitemOpen
  \bibfield  {author} {\bibinfo {author} {\bibfnamefont {R.~D.}\ \bibnamefont
  {King}}, \bibinfo {author} {\bibfnamefont {J.}~\bibnamefont {Rowland}},
  \bibinfo {author} {\bibfnamefont {S.~G.}\ \bibnamefont {Oliver}}, \bibinfo
  {author} {\bibfnamefont {M.}~\bibnamefont {Young}}, \bibinfo {author}
  {\bibfnamefont {W.}~\bibnamefont {Aubrey}}, \bibinfo {author} {\bibfnamefont
  {E.}~\bibnamefont {Byrne}}, \bibinfo {author} {\bibfnamefont
  {M.}~\bibnamefont {Liakata}}, \bibinfo {author} {\bibfnamefont
  {M.}~\bibnamefont {Markham}}, \bibinfo {author} {\bibfnamefont
  {P.}~\bibnamefont {Pir}}, \bibinfo {author} {\bibfnamefont {L.~N.}\
  \bibnamefont {Soldatova}}, \bibinfo {author} {\bibfnamefont {A.}~\bibnamefont
  {Sparkes}}, \bibinfo {author} {\bibfnamefont {K.~E.}\ \bibnamefont
  {Whelan}},\ and\ \bibinfo {author} {\bibfnamefont {A.}~\bibnamefont
  {Clare}},\ }\bibfield  {title} {\bibinfo {title} {The automation of
  science},\ }\href {https://doi.org/10.1126/science.1165620} {\bibfield
  {journal} {\bibinfo  {journal} {Science}\ }\textbf {\bibinfo {volume}
  {324}},\ \bibinfo {pages} {85} (\bibinfo {year} {2009})}\BibitemShut
  {NoStop}%
\bibitem [{\citenamefont {Schmidt}\ and\ \citenamefont
  {Lipson}(2009)}]{doi:10.1126/science.1165893}%
  \BibitemOpen
  \bibfield  {author} {\bibinfo {author} {\bibfnamefont {M.}~\bibnamefont
  {Schmidt}}\ and\ \bibinfo {author} {\bibfnamefont {H.}~\bibnamefont
  {Lipson}},\ }\bibfield  {title} {\bibinfo {title} {Distilling free-form
  natural laws from experimental data},\ }\href
  {https://doi.org/10.1126/science.1165893} {\bibfield  {journal} {\bibinfo
  {journal} {Science}\ }\textbf {\bibinfo {volume} {324}},\ \bibinfo {pages}
  {81} (\bibinfo {year} {2009})}\BibitemShut {NoStop}%
\bibitem [{\citenamefont {Wu}\ and\ \citenamefont
  {Tegmark}(2019)}]{PhysRevE.100.033311}%
  \BibitemOpen
  \bibfield  {author} {\bibinfo {author} {\bibfnamefont {T.}~\bibnamefont
  {Wu}}\ and\ \bibinfo {author} {\bibfnamefont {M.}~\bibnamefont {Tegmark}},\
  }\bibfield  {title} {\bibinfo {title} {Toward an artificial intelligence
  physicist for unsupervised learning},\ }\href
  {https://doi.org/10.1103/PhysRevE.100.033311} {\bibfield  {journal} {\bibinfo
   {journal} {Phys. Rev. E}\ }\textbf {\bibinfo {volume} {100}},\ \bibinfo
  {pages} {033311} (\bibinfo {year} {2019})}\BibitemShut {NoStop}%
\bibitem [{\citenamefont {Champion}\ \emph {et~al.}(2019)\citenamefont
  {Champion}, \citenamefont {Lusch}, \citenamefont {Kutz},\ and\ \citenamefont
  {Brunton}}]{doi:10.1073/pnas.1906995116}%
  \BibitemOpen
  \bibfield  {author} {\bibinfo {author} {\bibfnamefont {K.}~\bibnamefont
  {Champion}}, \bibinfo {author} {\bibfnamefont {B.}~\bibnamefont {Lusch}},
  \bibinfo {author} {\bibfnamefont {J.~N.}\ \bibnamefont {Kutz}},\ and\
  \bibinfo {author} {\bibfnamefont {S.~L.}\ \bibnamefont {Brunton}},\
  }\bibfield  {title} {\bibinfo {title} {Data-driven discovery of coordinates
  and governing equations},\ }\href {https://doi.org/10.1073/pnas.1906995116}
  {\bibfield  {journal} {\bibinfo  {journal} {Proceedings of the National
  Academy of Sciences}\ }\textbf {\bibinfo {volume} {116}},\ \bibinfo {pages}
  {22445} (\bibinfo {year} {2019})}\BibitemShut {NoStop}%
\bibitem [{\citenamefont {Iten}\ \emph {et~al.}(2020)\citenamefont {Iten},
  \citenamefont {Metger}, \citenamefont {Wilming}, \citenamefont {del Rio},\
  and\ \citenamefont {Renner}}]{PhysRevLett.124.010508}%
  \BibitemOpen
  \bibfield  {author} {\bibinfo {author} {\bibfnamefont {R.}~\bibnamefont
  {Iten}}, \bibinfo {author} {\bibfnamefont {T.}~\bibnamefont {Metger}},
  \bibinfo {author} {\bibfnamefont {H.}~\bibnamefont {Wilming}}, \bibinfo
  {author} {\bibfnamefont {L.}~\bibnamefont {del Rio}},\ and\ \bibinfo {author}
  {\bibfnamefont {R.}~\bibnamefont {Renner}},\ }\bibfield  {title} {\bibinfo
  {title} {Discovering physical concepts with neural networks},\ }\href
  {https://doi.org/10.1103/PhysRevLett.124.010508} {\bibfield  {journal}
  {\bibinfo  {journal} {Phys. Rev. Lett.}\ }\textbf {\bibinfo {volume} {124}},\
  \bibinfo {pages} {010508} (\bibinfo {year} {2020})}\BibitemShut {NoStop}%
\bibitem [{\citenamefont {Mototake}(2021)}]{PhysRevE.103.033303}%
  \BibitemOpen
  \bibfield  {author} {\bibinfo {author} {\bibfnamefont {Y.-i.}\ \bibnamefont
  {Mototake}},\ }\bibfield  {title} {\bibinfo {title} {Interpretable
  conservation law estimation by deriving the symmetries of dynamics from
  trained deep neural networks},\ }\href
  {https://doi.org/10.1103/PhysRevE.103.033303} {\bibfield  {journal} {\bibinfo
   {journal} {Phys. Rev. E}\ }\textbf {\bibinfo {volume} {103}},\ \bibinfo
  {pages} {033303} (\bibinfo {year} {2021})}\BibitemShut {NoStop}%
\bibitem [{\citenamefont {Wetzel}\ \emph {et~al.}(2020)\citenamefont {Wetzel},
  \citenamefont {Melko}, \citenamefont {Scott}, \citenamefont {Panju},\ and\
  \citenamefont {Ganesh}}]{PhysRevResearch.2.033499}%
  \BibitemOpen
  \bibfield  {author} {\bibinfo {author} {\bibfnamefont {S.~J.}\ \bibnamefont
  {Wetzel}}, \bibinfo {author} {\bibfnamefont {R.~G.}\ \bibnamefont {Melko}},
  \bibinfo {author} {\bibfnamefont {J.}~\bibnamefont {Scott}}, \bibinfo
  {author} {\bibfnamefont {M.}~\bibnamefont {Panju}},\ and\ \bibinfo {author}
  {\bibfnamefont {V.}~\bibnamefont {Ganesh}},\ }\bibfield  {title} {\bibinfo
  {title} {Discovering symmetry invariants and conserved quantities by
  interpreting siamese neural networks},\ }\href
  {https://doi.org/10.1103/PhysRevResearch.2.033499} {\bibfield  {journal}
  {\bibinfo  {journal} {Phys. Rev. Research}\ }\textbf {\bibinfo {volume}
  {2}},\ \bibinfo {pages} {033499} (\bibinfo {year} {2020})}\BibitemShut
  {NoStop}%
\bibitem [{\citenamefont {Alet}\ \emph {et~al.}(2021)\citenamefont {Alet},
  \citenamefont {Doblar}, \citenamefont {Zhou}, \citenamefont {Tenenbaum},
  \citenamefont {Kawaguchi},\ and\ \citenamefont
  {Finn}}]{NEURIPS2021_886ad506}%
  \BibitemOpen
  \bibfield  {author} {\bibinfo {author} {\bibfnamefont {F.}~\bibnamefont
  {Alet}}, \bibinfo {author} {\bibfnamefont {D.}~\bibnamefont {Doblar}},
  \bibinfo {author} {\bibfnamefont {A.}~\bibnamefont {Zhou}}, \bibinfo {author}
  {\bibfnamefont {J.}~\bibnamefont {Tenenbaum}}, \bibinfo {author}
  {\bibfnamefont {K.}~\bibnamefont {Kawaguchi}},\ and\ \bibinfo {author}
  {\bibfnamefont {C.}~\bibnamefont {Finn}},\ }\bibfield  {title} {\bibinfo
  {title} {Noether networks: meta-learning useful conserved quantities},\ }in\
  \href
  {https://proceedings.neurips.cc/paper/2021/file/886ad506e0c115cf590d18ebb6c26561-Paper.pdf}
  {\emph {\bibinfo {booktitle} {Advances in Neural Information Processing
  Systems}}},\ Vol.~\bibinfo {volume} {34},\ \bibinfo {editor} {edited by\
  \bibinfo {editor} {\bibfnamefont {M.}~\bibnamefont {Ranzato}}, \bibinfo
  {editor} {\bibfnamefont {A.}~\bibnamefont {Beygelzimer}}, \bibinfo {editor}
  {\bibfnamefont {Y.}~\bibnamefont {Dauphin}}, \bibinfo {editor} {\bibfnamefont
  {P.}~\bibnamefont {Liang}},\ and\ \bibinfo {editor} {\bibfnamefont {J.~W.}\
  \bibnamefont {Vaughan}}}\ (\bibinfo  {publisher} {Curran Associates, Inc.},\
  \bibinfo {year} {2021})\ pp.\ \bibinfo {pages} {16384--16397}\BibitemShut
  {NoStop}%
\bibitem [{\citenamefont {Kaiser}\ \emph {et~al.}(2018)\citenamefont {Kaiser},
  \citenamefont {Kutz},\ and\ \citenamefont {Brunton}}]{kaiser2018discovering}%
  \BibitemOpen
  \bibfield  {author} {\bibinfo {author} {\bibfnamefont {E.}~\bibnamefont
  {Kaiser}}, \bibinfo {author} {\bibfnamefont {J.~N.}\ \bibnamefont {Kutz}},\
  and\ \bibinfo {author} {\bibfnamefont {S.~L.}\ \bibnamefont {Brunton}},\
  }\bibfield  {title} {\bibinfo {title} {Discovering conservation laws from
  data for control},\ }in\ \href {https://ieeexplore.ieee.org/document/8618963}
  {\emph {\bibinfo {booktitle} {2018 IEEE Conference on Decision and Control
  (CDC)}}}\ (\bibinfo {organization} {IEEE},\ \bibinfo {year} {2018})\ pp.\
  \bibinfo {pages} {6415--6421}\BibitemShut {NoStop}%
\bibitem [{\citenamefont {Liu}\ and\ \citenamefont
  {Tegmark}(2021)}]{PhysRevLett.126.180604}%
  \BibitemOpen
  \bibfield  {author} {\bibinfo {author} {\bibfnamefont {Z.}~\bibnamefont
  {Liu}}\ and\ \bibinfo {author} {\bibfnamefont {M.}~\bibnamefont {Tegmark}},\
  }\bibfield  {title} {\bibinfo {title} {Machine learning conservation laws
  from trajectories},\ }\href {https://doi.org/10.1103/PhysRevLett.126.180604}
  {\bibfield  {journal} {\bibinfo  {journal} {Phys. Rev. Lett.}\ }\textbf
  {\bibinfo {volume} {126}},\ \bibinfo {pages} {180604} (\bibinfo {year}
  {2021})}\BibitemShut {NoStop}%
\bibitem [{\citenamefont {Ha}\ and\ \citenamefont
  {Jeong}(2021)}]{PhysRevResearch.3.L042035}%
  \BibitemOpen
  \bibfield  {author} {\bibinfo {author} {\bibfnamefont {S.}~\bibnamefont
  {Ha}}\ and\ \bibinfo {author} {\bibfnamefont {H.}~\bibnamefont {Jeong}},\
  }\bibfield  {title} {\bibinfo {title} {Discovering invariants via machine
  learning},\ }\href {https://doi.org/10.1103/PhysRevResearch.3.L042035}
  {\bibfield  {journal} {\bibinfo  {journal} {Phys. Rev. Research}\ }\textbf
  {\bibinfo {volume} {3}},\ \bibinfo {pages} {L042035} (\bibinfo {year}
  {2021})}\BibitemShut {NoStop}%
\bibitem [{\citenamefont {Calabrese}\ \emph {et~al.}(2016)\citenamefont
  {Calabrese}, \citenamefont {Essler},\ and\ \citenamefont
  {Mussardo}}]{Calabrese_2016}%
  \BibitemOpen
  \bibfield  {author} {\bibinfo {author} {\bibfnamefont {P.}~\bibnamefont
  {Calabrese}}, \bibinfo {author} {\bibfnamefont {F.~H.~L.}\ \bibnamefont
  {Essler}},\ and\ \bibinfo {author} {\bibfnamefont {G.}~\bibnamefont
  {Mussardo}},\ }\bibfield  {title} {\bibinfo {title} {Introduction to
  ‘quantum integrability in out of equilibrium systems’},\ }\href
  {https://doi.org/10.1088/1742-5468/2016/06/064001} {\bibfield  {journal}
  {\bibinfo  {journal} {Journal of Statistical Mechanics: Theory and
  Experiment}\ }\textbf {\bibinfo {volume} {2016}},\ \bibinfo {pages} {064001}
  (\bibinfo {year} {2016})}\BibitemShut {NoStop}%
\bibitem [{\citenamefont {Serbyn}\ \emph {et~al.}(2013)\citenamefont {Serbyn},
  \citenamefont {Papi\ifmmode~\acute{c}\else \'{c}\fi{}},\ and\ \citenamefont
  {Abanin}}]{PhysRevLett.111.127201}%
  \BibitemOpen
  \bibfield  {author} {\bibinfo {author} {\bibfnamefont {M.}~\bibnamefont
  {Serbyn}}, \bibinfo {author} {\bibfnamefont {Z.}~\bibnamefont
  {Papi\ifmmode~\acute{c}\else \'{c}\fi{}}},\ and\ \bibinfo {author}
  {\bibfnamefont {D.~A.}\ \bibnamefont {Abanin}},\ }\bibfield  {title}
  {\bibinfo {title} {Local conservation laws and the structure of the many-body
  localized states},\ }\href {https://doi.org/10.1103/PhysRevLett.111.127201}
  {\bibfield  {journal} {\bibinfo  {journal} {Phys. Rev. Lett.}\ }\textbf
  {\bibinfo {volume} {111}},\ \bibinfo {pages} {127201} (\bibinfo {year}
  {2013})}\BibitemShut {NoStop}%
\bibitem [{\citenamefont {Huse}\ \emph {et~al.}(2014)\citenamefont {Huse},
  \citenamefont {Nandkishore},\ and\ \citenamefont
  {Oganesyan}}]{PhysRevB.90.174202}%
  \BibitemOpen
  \bibfield  {author} {\bibinfo {author} {\bibfnamefont {D.~A.}\ \bibnamefont
  {Huse}}, \bibinfo {author} {\bibfnamefont {R.}~\bibnamefont {Nandkishore}},\
  and\ \bibinfo {author} {\bibfnamefont {V.}~\bibnamefont {Oganesyan}},\
  }\bibfield  {title} {\bibinfo {title} {Phenomenology of fully
  many-body-localized systems},\ }\href
  {https://doi.org/10.1103/PhysRevB.90.174202} {\bibfield  {journal} {\bibinfo
  {journal} {Phys. Rev. B}\ }\textbf {\bibinfo {volume} {90}},\ \bibinfo
  {pages} {174202} (\bibinfo {year} {2014})}\BibitemShut {NoStop}%
\bibitem [{\citenamefont {Ros}\ \emph {et~al.}(2015)\citenamefont {Ros},
  \citenamefont {Müller},\ and\ \citenamefont {Scardicchio}}]{ROS2015420}%
  \BibitemOpen
  \bibfield  {author} {\bibinfo {author} {\bibfnamefont {V.}~\bibnamefont
  {Ros}}, \bibinfo {author} {\bibfnamefont {M.}~\bibnamefont {Müller}},\ and\
  \bibinfo {author} {\bibfnamefont {A.}~\bibnamefont {Scardicchio}},\
  }\bibfield  {title} {\bibinfo {title} {Integrals of motion in the many-body
  localized phase},\ }\href
  {https://doi.org/https://doi.org/10.1016/j.nuclphysb.2014.12.014} {\bibfield
  {journal} {\bibinfo  {journal} {Nuclear Physics B}\ }\textbf {\bibinfo
  {volume} {891}},\ \bibinfo {pages} {420} (\bibinfo {year}
  {2015})}\BibitemShut {NoStop}%
\bibitem [{\citenamefont {Imbrie}\ \emph {et~al.}(2017)\citenamefont {Imbrie},
  \citenamefont {Ros},\ and\ \citenamefont {Scardicchio}}]{imbrie2017local}%
  \BibitemOpen
  \bibfield  {author} {\bibinfo {author} {\bibfnamefont {J.~Z.}\ \bibnamefont
  {Imbrie}}, \bibinfo {author} {\bibfnamefont {V.}~\bibnamefont {Ros}},\ and\
  \bibinfo {author} {\bibfnamefont {A.}~\bibnamefont {Scardicchio}},\
  }\bibfield  {title} {\bibinfo {title} {Local integrals of motion in many-body
  localized systems},\ }\href {https://doi.org/10.1002/andp.201600278}
  {\bibfield  {journal} {\bibinfo  {journal} {Annalen der Physik}\ }\textbf
  {\bibinfo {volume} {529}},\ \bibinfo {pages} {1600278} (\bibinfo {year}
  {2017})}\BibitemShut {NoStop}%
\bibitem [{\citenamefont {Anderson}(1958)}]{PhysRev.109.1492}%
  \BibitemOpen
  \bibfield  {author} {\bibinfo {author} {\bibfnamefont {P.~W.}\ \bibnamefont
  {Anderson}},\ }\bibfield  {title} {\bibinfo {title} {Absence of diffusion in
  certain random lattices},\ }\href {https://doi.org/10.1103/PhysRev.109.1492}
  {\bibfield  {journal} {\bibinfo  {journal} {Phys. Rev.}\ }\textbf {\bibinfo
  {volume} {109}},\ \bibinfo {pages} {1492} (\bibinfo {year}
  {1958})}\BibitemShut {NoStop}%
\bibitem [{\citenamefont {Fleishman}\ and\ \citenamefont
  {Anderson}(1980)}]{PhysRevB.21.2366}%
  \BibitemOpen
  \bibfield  {author} {\bibinfo {author} {\bibfnamefont {L.}~\bibnamefont
  {Fleishman}}\ and\ \bibinfo {author} {\bibfnamefont {P.~W.}\ \bibnamefont
  {Anderson}},\ }\bibfield  {title} {\bibinfo {title} {Interactions and the
  anderson transition},\ }\href {https://doi.org/10.1103/PhysRevB.21.2366}
  {\bibfield  {journal} {\bibinfo  {journal} {Phys. Rev. B}\ }\textbf {\bibinfo
  {volume} {21}},\ \bibinfo {pages} {2366} (\bibinfo {year}
  {1980})}\BibitemShut {NoStop}%
\bibitem [{\citenamefont {Gornyi}\ \emph {et~al.}(2005)\citenamefont {Gornyi},
  \citenamefont {Mirlin},\ and\ \citenamefont
  {Polyakov}}]{PhysRevLett.95.206603}%
  \BibitemOpen
  \bibfield  {author} {\bibinfo {author} {\bibfnamefont {I.~V.}\ \bibnamefont
  {Gornyi}}, \bibinfo {author} {\bibfnamefont {A.~D.}\ \bibnamefont {Mirlin}},\
  and\ \bibinfo {author} {\bibfnamefont {D.~G.}\ \bibnamefont {Polyakov}},\
  }\bibfield  {title} {\bibinfo {title} {Interacting electrons in disordered
  wires: Anderson localization and low-$t$ transport},\ }\href
  {https://doi.org/10.1103/PhysRevLett.95.206603} {\bibfield  {journal}
  {\bibinfo  {journal} {Phys. Rev. Lett.}\ }\textbf {\bibinfo {volume} {95}},\
  \bibinfo {pages} {206603} (\bibinfo {year} {2005})}\BibitemShut {NoStop}%
\bibitem [{\citenamefont {Basko}\ \emph {et~al.}(2006)\citenamefont {Basko},
  \citenamefont {Aleiner},\ and\ \citenamefont {Altshuler}}]{BASKO20061126}%
  \BibitemOpen
  \bibfield  {author} {\bibinfo {author} {\bibfnamefont {D.}~\bibnamefont
  {Basko}}, \bibinfo {author} {\bibfnamefont {I.}~\bibnamefont {Aleiner}},\
  and\ \bibinfo {author} {\bibfnamefont {B.}~\bibnamefont {Altshuler}},\
  }\bibfield  {title} {\bibinfo {title} {Metal–insulator transition in a
  weakly interacting many-electron system with localized single-particle
  states},\ }\href {https://doi.org/https://doi.org/10.1016/j.aop.2005.11.014}
  {\bibfield  {journal} {\bibinfo  {journal} {Annals of Physics}\ }\textbf
  {\bibinfo {volume} {321}},\ \bibinfo {pages} {1126} (\bibinfo {year}
  {2006})}\BibitemShut {NoStop}%
\bibitem [{\citenamefont {Abanin}\ and\ \citenamefont
  {Papić}(2017)}]{https://doi.org/10.1002/andp.201700169}%
  \BibitemOpen
  \bibfield  {author} {\bibinfo {author} {\bibfnamefont {D.~A.}\ \bibnamefont
  {Abanin}}\ and\ \bibinfo {author} {\bibfnamefont {Z.}~\bibnamefont
  {Papić}},\ }\bibfield  {title} {\bibinfo {title} {Recent progress in
  many-body localization},\ }\href
  {https://doi.org/https://doi.org/10.1002/andp.201700169} {\bibfield
  {journal} {\bibinfo  {journal} {Annalen der Physik}\ }\textbf {\bibinfo
  {volume} {529}},\ \bibinfo {pages} {1700169} (\bibinfo {year}
  {2017})}\BibitemShut {NoStop}%
\bibitem [{\citenamefont {Nandkishore}\ and\ \citenamefont
  {Huse}(2015)}]{doi:10.1146/annurev-conmatphys-031214-014726}%
  \BibitemOpen
  \bibfield  {author} {\bibinfo {author} {\bibfnamefont {R.}~\bibnamefont
  {Nandkishore}}\ and\ \bibinfo {author} {\bibfnamefont {D.~A.}\ \bibnamefont
  {Huse}},\ }\bibfield  {title} {\bibinfo {title} {Many-body localization and
  thermalization in quantum statistical mechanics},\ }\href
  {https://doi.org/10.1146/annurev-conmatphys-031214-014726} {\bibfield
  {journal} {\bibinfo  {journal} {Annual Review of Condensed Matter Physics}\
  }\textbf {\bibinfo {volume} {6}},\ \bibinfo {pages} {15} (\bibinfo {year}
  {2015})}\BibitemShut {NoStop}%
\bibitem [{\citenamefont {Abanin}\ \emph {et~al.}(2019)\citenamefont {Abanin},
  \citenamefont {Altman}, \citenamefont {Bloch},\ and\ \citenamefont
  {Serbyn}}]{RevModPhys.91.021001}%
  \BibitemOpen
  \bibfield  {author} {\bibinfo {author} {\bibfnamefont {D.~A.}\ \bibnamefont
  {Abanin}}, \bibinfo {author} {\bibfnamefont {E.}~\bibnamefont {Altman}},
  \bibinfo {author} {\bibfnamefont {I.}~\bibnamefont {Bloch}},\ and\ \bibinfo
  {author} {\bibfnamefont {M.}~\bibnamefont {Serbyn}},\ }\bibfield  {title}
  {\bibinfo {title} {Colloquium: Many-body localization, thermalization, and
  entanglement},\ }\href {https://doi.org/10.1103/RevModPhys.91.021001}
  {\bibfield  {journal} {\bibinfo  {journal} {Rev. Mod. Phys.}\ }\textbf
  {\bibinfo {volume} {91}},\ \bibinfo {pages} {021001} (\bibinfo {year}
  {2019})}\BibitemShut {NoStop}%
\bibitem [{\citenamefont {Moudgalya}\ \emph {et~al.}(2022)\citenamefont
  {Moudgalya}, \citenamefont {Bernevig},\ and\ \citenamefont
  {Regnault}}]{moudgalya2022quantum}%
  \BibitemOpen
  \bibfield  {author} {\bibinfo {author} {\bibfnamefont {S.}~\bibnamefont
  {Moudgalya}}, \bibinfo {author} {\bibfnamefont {B.~A.}\ \bibnamefont
  {Bernevig}},\ and\ \bibinfo {author} {\bibfnamefont {N.}~\bibnamefont
  {Regnault}},\ }\bibfield  {title} {\bibinfo {title} {Quantum many-body scars
  and hilbert space fragmentation: a review of exact results},\ }\href
  {https://doi.org/10.1088%2F1361-6633%2Fac73a0} {\bibfield  {journal}
  {\bibinfo  {journal} {Reports on Progress in Physics}\ }\textbf {\bibinfo
  {volume} {85}},\ \bibinfo {pages} {086501} (\bibinfo {year}
  {2022})}\BibitemShut {NoStop}%
\bibitem [{\citenamefont {Huang}\ \emph {et~al.}(2020)\citenamefont {Huang},
  \citenamefont {Kueng},\ and\ \citenamefont {Preskill}}]{huang2020predicting}%
  \BibitemOpen
  \bibfield  {author} {\bibinfo {author} {\bibfnamefont {H.-Y.}\ \bibnamefont
  {Huang}}, \bibinfo {author} {\bibfnamefont {R.}~\bibnamefont {Kueng}},\ and\
  \bibinfo {author} {\bibfnamefont {J.}~\bibnamefont {Preskill}},\ }\bibfield
  {title} {\bibinfo {title} {Predicting many properties of a quantum system
  from very few measurements},\ }\href
  {https://doi.org/10.1038/s41567-020-0932-7} {\bibfield  {journal} {\bibinfo
  {journal} {Nature Physics}\ }\textbf {\bibinfo {volume} {16}},\ \bibinfo
  {pages} {1050} (\bibinfo {year} {2020})}\BibitemShut {NoStop}%
\bibitem [{\citenamefont {Elben}\ \emph {et~al.}(2022)\citenamefont {Elben},
  \citenamefont {Flammia}, \citenamefont {Huang}, \citenamefont {Kueng},
  \citenamefont {Preskill}, \citenamefont {Vermersch},\ and\ \citenamefont
  {Zoller}}]{elben2022randomized}%
  \BibitemOpen
  \bibfield  {author} {\bibinfo {author} {\bibfnamefont {A.}~\bibnamefont
  {Elben}}, \bibinfo {author} {\bibfnamefont {S.~T.}\ \bibnamefont {Flammia}},
  \bibinfo {author} {\bibfnamefont {H.-Y.}\ \bibnamefont {Huang}}, \bibinfo
  {author} {\bibfnamefont {R.}~\bibnamefont {Kueng}}, \bibinfo {author}
  {\bibfnamefont {J.}~\bibnamefont {Preskill}}, \bibinfo {author}
  {\bibfnamefont {B.}~\bibnamefont {Vermersch}},\ and\ \bibinfo {author}
  {\bibfnamefont {P.}~\bibnamefont {Zoller}},\ }\bibfield  {title} {\bibinfo
  {title} {The randomized measurement toolbox},\ }\href
  {https://doi.org/10.1038%2Fs42254-022-00535-2} {\bibfield  {journal}
  {\bibinfo  {journal} {Nature Reviews Physics}\ }\textbf {\bibinfo {volume}
  {5}},\ \bibinfo {pages} {9} (\bibinfo {year} {2022})}\BibitemShut {NoStop}%
\bibitem [{\citenamefont {Chandran}\ \emph {et~al.}(2015)\citenamefont
  {Chandran}, \citenamefont {Kim}, \citenamefont {Vidal},\ and\ \citenamefont
  {Abanin}}]{PhysRevB.91.085425}%
  \BibitemOpen
  \bibfield  {author} {\bibinfo {author} {\bibfnamefont {A.}~\bibnamefont
  {Chandran}}, \bibinfo {author} {\bibfnamefont {I.~H.}\ \bibnamefont {Kim}},
  \bibinfo {author} {\bibfnamefont {G.}~\bibnamefont {Vidal}},\ and\ \bibinfo
  {author} {\bibfnamefont {D.~A.}\ \bibnamefont {Abanin}},\ }\bibfield  {title}
  {\bibinfo {title} {Constructing local integrals of motion in the many-body
  localized phase},\ }\href {https://doi.org/10.1103/PhysRevB.91.085425}
  {\bibfield  {journal} {\bibinfo  {journal} {Phys. Rev. B}\ }\textbf {\bibinfo
  {volume} {91}},\ \bibinfo {pages} {085425} (\bibinfo {year}
  {2015})}\BibitemShut {NoStop}%
\bibitem [{\citenamefont {Mierzejewski}\ \emph {et~al.}(2015)\citenamefont
  {Mierzejewski}, \citenamefont {Prelov\ifmmode~\check{s}\else \v{s}\fi{}ek},\
  and\ \citenamefont {Prosen}}]{Mierzejewski2015identify}%
  \BibitemOpen
  \bibfield  {author} {\bibinfo {author} {\bibfnamefont {M.}~\bibnamefont
  {Mierzejewski}}, \bibinfo {author} {\bibfnamefont {P.}~\bibnamefont
  {Prelov\ifmmode~\check{s}\else \v{s}\fi{}ek}},\ and\ \bibinfo {author}
  {\bibfnamefont {T.~c.~v.}\ \bibnamefont {Prosen}},\ }\bibfield  {title}
  {\bibinfo {title} {Identifying local and quasilocal conserved quantities in
  integrable systems},\ }\href
  {https://link.aps.org/doi/10.1103/PhysRevLett.114.140601} {\bibfield
  {journal} {\bibinfo  {journal} {Phys. Rev. Lett.}\ }\textbf {\bibinfo
  {volume} {114}},\ \bibinfo {pages} {140601} (\bibinfo {year}
  {2015})}\BibitemShut {NoStop}%
\bibitem [{\citenamefont {O'Brien}\ \emph {et~al.}(2016)\citenamefont
  {O'Brien}, \citenamefont {Abanin}, \citenamefont {Vidal},\ and\ \citenamefont
  {Papi\ifmmode~\acute{c}\else \'{c}\fi{}}}]{PhysRevB.94.144208}%
  \BibitemOpen
  \bibfield  {author} {\bibinfo {author} {\bibfnamefont {T.~E.}\ \bibnamefont
  {O'Brien}}, \bibinfo {author} {\bibfnamefont {D.~A.}\ \bibnamefont {Abanin}},
  \bibinfo {author} {\bibfnamefont {G.}~\bibnamefont {Vidal}},\ and\ \bibinfo
  {author} {\bibfnamefont {Z.}~\bibnamefont {Papi\ifmmode~\acute{c}\else
  \'{c}\fi{}}},\ }\bibfield  {title} {\bibinfo {title} {Explicit construction
  of local conserved operators in disordered many-body systems},\ }\href
  {https://doi.org/10.1103/PhysRevB.94.144208} {\bibfield  {journal} {\bibinfo
  {journal} {Phys. Rev. B}\ }\textbf {\bibinfo {volume} {94}},\ \bibinfo
  {pages} {144208} (\bibinfo {year} {2016})}\BibitemShut {NoStop}%
\bibitem [{\citenamefont {Chertkov}\ \emph {et~al.}(2021)\citenamefont
  {Chertkov}, \citenamefont {Villalonga},\ and\ \citenamefont
  {Clark}}]{PhysRevLett.126.180602}%
  \BibitemOpen
  \bibfield  {author} {\bibinfo {author} {\bibfnamefont {E.}~\bibnamefont
  {Chertkov}}, \bibinfo {author} {\bibfnamefont {B.}~\bibnamefont
  {Villalonga}},\ and\ \bibinfo {author} {\bibfnamefont {B.~K.}\ \bibnamefont
  {Clark}},\ }\bibfield  {title} {\bibinfo {title} {Numerical evidence for
  many-body localization in two and three dimensions},\ }\href
  {https://doi.org/10.1103/PhysRevLett.126.180602} {\bibfield  {journal}
  {\bibinfo  {journal} {Phys. Rev. Lett.}\ }\textbf {\bibinfo {volume} {126}},\
  \bibinfo {pages} {180602} (\bibinfo {year} {2021})}\BibitemShut {NoStop}%
\bibitem [{\citenamefont {Bentsen}\ \emph {et~al.}(2019)\citenamefont
  {Bentsen}, \citenamefont {Potirniche}, \citenamefont {Bulchandani},
  \citenamefont {Scaffidi}, \citenamefont {Cao}, \citenamefont {Qi},
  \citenamefont {Schleier-Smith},\ and\ \citenamefont
  {Altman}}]{Bentsen2019integrable}%
  \BibitemOpen
  \bibfield  {author} {\bibinfo {author} {\bibfnamefont {G.}~\bibnamefont
  {Bentsen}}, \bibinfo {author} {\bibfnamefont {I.-D.}\ \bibnamefont
  {Potirniche}}, \bibinfo {author} {\bibfnamefont {V.~B.}\ \bibnamefont
  {Bulchandani}}, \bibinfo {author} {\bibfnamefont {T.}~\bibnamefont
  {Scaffidi}}, \bibinfo {author} {\bibfnamefont {X.}~\bibnamefont {Cao}},
  \bibinfo {author} {\bibfnamefont {X.-L.}\ \bibnamefont {Qi}}, \bibinfo
  {author} {\bibfnamefont {M.}~\bibnamefont {Schleier-Smith}},\ and\ \bibinfo
  {author} {\bibfnamefont {E.}~\bibnamefont {Altman}},\ }\bibfield  {title}
  {\bibinfo {title} {Integrable and chaotic dynamics of spins coupled to an
  optical cavity},\ }\href {https://doi.org/10.1103%2Fphysrevx.9.041011}
  {\bibfield  {journal} {\bibinfo  {journal} {Physical Review X}\ }\textbf
  {\bibinfo {volume} {9}} (\bibinfo {year} {2019})}\BibitemShut {NoStop}%
\bibitem [{\citenamefont {Wang}\ \emph {et~al.}(2017)\citenamefont {Wang},
  \citenamefont {Paesani}, \citenamefont {Santagati}, \citenamefont {Knauer},
  \citenamefont {Gentile}, \citenamefont {Wiebe}, \citenamefont {Petruzzella},
  \citenamefont {O’Brien}, \citenamefont {Rarity}, \citenamefont {Laing}
  \emph {et~al.}}]{wang2017experimental}%
  \BibitemOpen
  \bibfield  {author} {\bibinfo {author} {\bibfnamefont {J.}~\bibnamefont
  {Wang}}, \bibinfo {author} {\bibfnamefont {S.}~\bibnamefont {Paesani}},
  \bibinfo {author} {\bibfnamefont {R.}~\bibnamefont {Santagati}}, \bibinfo
  {author} {\bibfnamefont {S.}~\bibnamefont {Knauer}}, \bibinfo {author}
  {\bibfnamefont {A.~A.}\ \bibnamefont {Gentile}}, \bibinfo {author}
  {\bibfnamefont {N.}~\bibnamefont {Wiebe}}, \bibinfo {author} {\bibfnamefont
  {M.}~\bibnamefont {Petruzzella}}, \bibinfo {author} {\bibfnamefont {J.~L.}\
  \bibnamefont {O’Brien}}, \bibinfo {author} {\bibfnamefont {J.~G.}\
  \bibnamefont {Rarity}}, \bibinfo {author} {\bibfnamefont {A.}~\bibnamefont
  {Laing}}, \emph {et~al.},\ }\bibfield  {title} {\bibinfo {title}
  {Experimental quantum hamiltonian learning},\ }\href
  {https://doi.org/10.1038/nphys4074} {\bibfield  {journal} {\bibinfo
  {journal} {Nat. Phys.}\ }\textbf {\bibinfo {volume} {13}},\ \bibinfo {pages}
  {551} (\bibinfo {year} {2017})}\BibitemShut {NoStop}%
\bibitem [{\citenamefont {Evans}\ \emph {et~al.}()\citenamefont {Evans},
  \citenamefont {Harper},\ and\ \citenamefont {Flammia}}]{evans2019scalable}%
  \BibitemOpen
  \bibfield  {author} {\bibinfo {author} {\bibfnamefont {T.~J.}\ \bibnamefont
  {Evans}}, \bibinfo {author} {\bibfnamefont {R.}~\bibnamefont {Harper}},\ and\
  \bibinfo {author} {\bibfnamefont {S.~T.}\ \bibnamefont {Flammia}},\
  }\href@noop {} {\bibinfo {title} {Scalable bayesian {Hamiltonian}
  learning}},\ \Eprint {https://arxiv.org/abs/1912.07636} {arXiv:1912.07636}
  \BibitemShut {NoStop}%
\bibitem [{\citenamefont {Gu}\ \emph {et~al.}()\citenamefont {Gu},
  \citenamefont {Cincio},\ and\ \citenamefont {Coles}}]{gu2022practical}%
  \BibitemOpen
  \bibfield  {author} {\bibinfo {author} {\bibfnamefont {A.}~\bibnamefont
  {Gu}}, \bibinfo {author} {\bibfnamefont {L.}~\bibnamefont {Cincio}},\ and\
  \bibinfo {author} {\bibfnamefont {P.~J.}\ \bibnamefont {Coles}},\ }\href@noop
  {} {\bibinfo {title} {Practical black box hamiltonian learning}},\ \Eprint
  {https://arxiv.org/abs/2206.15464} {arXiv:2206.15464} \BibitemShut {NoStop}%
\bibitem [{\citenamefont {Granade}\ \emph
  {et~al.}(2012{\natexlab{a}})\citenamefont {Granade}, \citenamefont {Ferrie},
  \citenamefont {Wiebe},\ and\ \citenamefont {Cory}}]{granade2012}%
  \BibitemOpen
  \bibfield  {author} {\bibinfo {author} {\bibfnamefont {C.~E.}\ \bibnamefont
  {Granade}}, \bibinfo {author} {\bibfnamefont {C.}~\bibnamefont {Ferrie}},
  \bibinfo {author} {\bibfnamefont {N.}~\bibnamefont {Wiebe}},\ and\ \bibinfo
  {author} {\bibfnamefont {D.~G.}\ \bibnamefont {Cory}},\ }\bibfield  {title}
  {\bibinfo {title} {Robust online hamiltonian learning},\ }\href
  {https://doi.org/10.1088%2F1367-2630%2F14%2F10%2F103013} {\bibfield
  {journal} {\bibinfo  {journal} {New Journal of Physics}\ }\textbf {\bibinfo
  {volume} {14}} (\bibinfo {year} {2012}{\natexlab{a}})}\BibitemShut {NoStop}%
\bibitem [{\citenamefont {Hangleiter}\ \emph {et~al.}(2021)\citenamefont
  {Hangleiter}, \citenamefont {Roth}, \citenamefont {Eisert},\ and\
  \citenamefont {Roushan}}]{hangleiter2021}%
  \BibitemOpen
  \bibfield  {author} {\bibinfo {author} {\bibfnamefont {D.}~\bibnamefont
  {Hangleiter}}, \bibinfo {author} {\bibfnamefont {I.}~\bibnamefont {Roth}},
  \bibinfo {author} {\bibfnamefont {J.}~\bibnamefont {Eisert}},\ and\ \bibinfo
  {author} {\bibfnamefont {P.}~\bibnamefont {Roushan}},\ }\href@noop {}
  {\bibinfo {title} {Precise hamiltonian identification of a superconducting
  quantum processor}} (\bibinfo {year} {2021}),\ \Eprint
  {https://arxiv.org/abs/2108.08319} {arXiv:2108.08319 [quant-ph]} \BibitemShut
  {NoStop}%
\bibitem [{\citenamefont {Wiebe}\ \emph
  {et~al.}(2014{\natexlab{a}})\citenamefont {Wiebe}, \citenamefont {Granade},
  \citenamefont {Ferrie},\ and\ \citenamefont {Cory}}]{wiebe2014a}%
  \BibitemOpen
  \bibfield  {author} {\bibinfo {author} {\bibfnamefont {N.}~\bibnamefont
  {Wiebe}}, \bibinfo {author} {\bibfnamefont {C.}~\bibnamefont {Granade}},
  \bibinfo {author} {\bibfnamefont {C.}~\bibnamefont {Ferrie}},\ and\ \bibinfo
  {author} {\bibfnamefont {D.}~\bibnamefont {Cory}},\ }\bibfield  {title}
  {\bibinfo {title} {Quantum hamiltonian learning using imperfect quantum
  resources},\ }\href {https://link.aps.org/doi/10.1103/PhysRevA.89.042314}
  {\bibfield  {journal} {\bibinfo  {journal} {Physical Review A}\ }\textbf
  {\bibinfo {volume} {89}} (\bibinfo {year} {2014}{\natexlab{a}})}\BibitemShut
  {NoStop}%
\bibitem [{\citenamefont {Wiebe}\ \emph
  {et~al.}(2014{\natexlab{b}})\citenamefont {Wiebe}, \citenamefont {Granade},
  \citenamefont {Ferrie},\ and\ \citenamefont {Cory}}]{wiebe2014b}%
  \BibitemOpen
  \bibfield  {author} {\bibinfo {author} {\bibfnamefont {N.}~\bibnamefont
  {Wiebe}}, \bibinfo {author} {\bibfnamefont {C.}~\bibnamefont {Granade}},
  \bibinfo {author} {\bibfnamefont {C.}~\bibnamefont {Ferrie}},\ and\ \bibinfo
  {author} {\bibfnamefont {D.~G.}\ \bibnamefont {Cory}},\ }\bibfield  {title}
  {\bibinfo {title} {Hamiltonian learning and certification using quantum
  resources},\ }\href {https://link.aps.org/doi/10.1103/PhysRevLett.112.190501}
  {\bibfield  {journal} {\bibinfo  {journal} {Physical Review Letters}\
  }\textbf {\bibinfo {volume} {112}} (\bibinfo {year}
  {2014}{\natexlab{b}})}\BibitemShut {NoStop}%
\bibitem [{\citenamefont {Yu}\ \emph {et~al.}(2022)\citenamefont {Yu},
  \citenamefont {Sun}, \citenamefont {Han},\ and\ \citenamefont
  {Yuan}}]{yu2022}%
  \BibitemOpen
  \bibfield  {author} {\bibinfo {author} {\bibfnamefont {W.}~\bibnamefont
  {Yu}}, \bibinfo {author} {\bibfnamefont {J.}~\bibnamefont {Sun}}, \bibinfo
  {author} {\bibfnamefont {Z.}~\bibnamefont {Han}},\ and\ \bibinfo {author}
  {\bibfnamefont {X.}~\bibnamefont {Yuan}},\ }\href@noop {} {\bibinfo {title}
  {Practical and efficient hamiltonian learning}} (\bibinfo {year} {2022}),\
  \Eprint {https://arxiv.org/abs/2201.00190} {arXiv:2201.00190 [quant-ph]}
  \BibitemShut {NoStop}%
\bibitem [{\citenamefont {Zubida}\ \emph {et~al.}()\citenamefont {Zubida},
  \citenamefont {Yitzhaki}, \citenamefont {Lindner},\ and\ \citenamefont
  {Bairey}}]{ZubidaYitzhakiEtAl2021optimal}%
  \BibitemOpen
  \bibfield  {author} {\bibinfo {author} {\bibfnamefont {A.}~\bibnamefont
  {Zubida}}, \bibinfo {author} {\bibfnamefont {E.}~\bibnamefont {Yitzhaki}},
  \bibinfo {author} {\bibfnamefont {N.~H.}\ \bibnamefont {Lindner}},\ and\
  \bibinfo {author} {\bibfnamefont {E.}~\bibnamefont {Bairey}},\ }\href@noop {}
  {\bibinfo {title} {Optimal short-time measurements for hamiltonian
  learning}},\ \Eprint {https://arxiv.org/abs/2108.08824} {arXiv:2108.08824}
  \BibitemShut {NoStop}%
\bibitem [{\citenamefont {Granade}\ \emph
  {et~al.}(2012{\natexlab{b}})\citenamefont {Granade}, \citenamefont {Ferrie},
  \citenamefont {Wiebe},\ and\ \citenamefont
  {Cory}}]{GranadeFerrieWiebeCory2012robust}%
  \BibitemOpen
  \bibfield  {author} {\bibinfo {author} {\bibfnamefont {C.~E.}\ \bibnamefont
  {Granade}}, \bibinfo {author} {\bibfnamefont {C.}~\bibnamefont {Ferrie}},
  \bibinfo {author} {\bibfnamefont {N.}~\bibnamefont {Wiebe}},\ and\ \bibinfo
  {author} {\bibfnamefont {D.~G.}\ \bibnamefont {Cory}},\ }\bibfield  {title}
  {\bibinfo {title} {Robust online hamiltonian learning},\ }\href
  {https://doi.org/10.1088/1367-2630/14/10/103013} {\bibfield  {journal}
  {\bibinfo  {journal} {New Journal of Physics}\ }\textbf {\bibinfo {volume}
  {14}},\ \bibinfo {pages} {103013} (\bibinfo {year}
  {2012}{\natexlab{b}})}\BibitemShut {NoStop}%
\bibitem [{\citenamefont {Li}\ \emph {et~al.}(2020)\citenamefont {Li},
  \citenamefont {Zou},\ and\ \citenamefont {Hsieh}}]{li2020hamiltonian}%
  \BibitemOpen
  \bibfield  {author} {\bibinfo {author} {\bibfnamefont {Z.}~\bibnamefont
  {Li}}, \bibinfo {author} {\bibfnamefont {L.}~\bibnamefont {Zou}},\ and\
  \bibinfo {author} {\bibfnamefont {T.~H.}\ \bibnamefont {Hsieh}},\ }\bibfield
  {title} {\bibinfo {title} {Hamiltonian tomography via quantum quench},\
  }\href {https://doi.org/10.1103%2Fphysrevlett.124.160502} {\bibfield
  {journal} {\bibinfo  {journal} {Physical Review Letters}\ }\textbf {\bibinfo
  {volume} {124}} (\bibinfo {year} {2020})}\BibitemShut {NoStop}%
\bibitem [{\citenamefont {Franca}\ \emph {et~al.}()\citenamefont {Franca},
  \citenamefont {Markovich}, \citenamefont {Dobrovitski}, \citenamefont
  {Werner},\ and\ \citenamefont {Borregaard}}]{FrancaEtAl2022efficient}%
  \BibitemOpen
  \bibfield  {author} {\bibinfo {author} {\bibfnamefont {D.~S.}\ \bibnamefont
  {Franca}}, \bibinfo {author} {\bibfnamefont {L.~A.}\ \bibnamefont
  {Markovich}}, \bibinfo {author} {\bibfnamefont {V.}~\bibnamefont
  {Dobrovitski}}, \bibinfo {author} {\bibfnamefont {A.~H.}\ \bibnamefont
  {Werner}},\ and\ \bibinfo {author} {\bibfnamefont {J.}~\bibnamefont
  {Borregaard}},\ }\href@noop {} {\bibinfo {title} {Efficient and robust
  estimation of many-qubit hamiltonians}},\ \Eprint
  {https://arxiv.org/abs/2205.09567} {arXiv:2205.09567} \BibitemShut {NoStop}%
\bibitem [{\citenamefont {Haah}\ \emph {et~al.}(2021)\citenamefont {Haah},
  \citenamefont {Kothari},\ and\ \citenamefont {Tang}}]{tang2021}%
  \BibitemOpen
  \bibfield  {author} {\bibinfo {author} {\bibfnamefont {J.}~\bibnamefont
  {Haah}}, \bibinfo {author} {\bibfnamefont {R.}~\bibnamefont {Kothari}},\ and\
  \bibinfo {author} {\bibfnamefont {E.}~\bibnamefont {Tang}},\ }\href@noop {}
  {\bibinfo {title} {Optimal learning of quantum hamiltonians from
  high-temperature gibbs states}} (\bibinfo {year} {2021}),\ \Eprint
  {https://arxiv.org/abs/2108.04842} {arXiv:2108.04842 [quant-ph]} \BibitemShut
  {NoStop}%
\bibitem [{\citenamefont {Pastori}\ \emph {et~al.}(2022)\citenamefont
  {Pastori}, \citenamefont {Olsacher}, \citenamefont {Kokail},\ and\
  \citenamefont {Zoller}}]{pastori2022characterization}%
  \BibitemOpen
  \bibfield  {author} {\bibinfo {author} {\bibfnamefont {L.}~\bibnamefont
  {Pastori}}, \bibinfo {author} {\bibfnamefont {T.}~\bibnamefont {Olsacher}},
  \bibinfo {author} {\bibfnamefont {C.}~\bibnamefont {Kokail}},\ and\ \bibinfo
  {author} {\bibfnamefont {P.}~\bibnamefont {Zoller}},\ }\bibfield  {title}
  {\bibinfo {title} {Characterization and verification of trotterized digital
  quantum simulation via hamiltonian and liouvillian learning},\ }\href
  {https://doi.org/10.1103%2Fprxquantum.3.030324} {\bibfield  {journal}
  {\bibinfo  {journal} {{PRX} Quantum}\ }\textbf {\bibinfo {volume} {3}}
  (\bibinfo {year} {2022})}\BibitemShut {NoStop}%
\bibitem [{\citenamefont {Caro}()}]{caro2022learning}%
  \BibitemOpen
  \bibfield  {author} {\bibinfo {author} {\bibfnamefont {M.~C.}\ \bibnamefont
  {Caro}},\ }\href@noop {} {\bibinfo {title} {Learning quantum processes and
  hamiltonians via the pauli transfer matrix}},\ \Eprint
  {https://arxiv.org/abs/2212.04471} {arXiv:2212.04471} \BibitemShut {NoStop}%
\bibitem [{\citenamefont {Huang}\ \emph {et~al.}()\citenamefont {Huang},
  \citenamefont {Tong}, \citenamefont {Fang},\ and\ \citenamefont
  {Su}}]{huang2022learning}%
  \BibitemOpen
  \bibfield  {author} {\bibinfo {author} {\bibfnamefont {H.-Y.}\ \bibnamefont
  {Huang}}, \bibinfo {author} {\bibfnamefont {Y.}~\bibnamefont {Tong}},
  \bibinfo {author} {\bibfnamefont {D.}~\bibnamefont {Fang}},\ and\ \bibinfo
  {author} {\bibfnamefont {Y.}~\bibnamefont {Su}},\ }\href@noop {} {\bibinfo
  {title} {Learning many-body hamiltonians with heisenberg-limited scaling}},\
  \Eprint {https://arxiv.org/abs/2210.03030} {arXiv:2210.03030} \BibitemShut
  {NoStop}%
\bibitem [{\citenamefont {Zhao}\ \emph {et~al.}(2023)\citenamefont {Zhao},
  \citenamefont {Hu},\ and\ \citenamefont {Zhang}}]{zhao2022supervised}%
  \BibitemOpen
  \bibfield  {author} {\bibinfo {author} {\bibfnamefont {T.-L.}\ \bibnamefont
  {Zhao}}, \bibinfo {author} {\bibfnamefont {S.-X.}\ \bibnamefont {Hu}},\ and\
  \bibinfo {author} {\bibfnamefont {Y.}~\bibnamefont {Zhang}},\ }\bibfield
  {title} {\bibinfo {title} {Maximum-likelihood-estimate hamiltonian learning
  via efficient and robust quantum likelihood gradient},\ }\href
  {https://doi.org/10.1103%2Fphysrevresearch.5.023136} {\bibfield  {journal}
  {\bibinfo  {journal} {Physical Review Research}\ }\textbf {\bibinfo {volume}
  {5}} (\bibinfo {year} {2023})}\BibitemShut {NoStop}%
\bibitem [{\citenamefont {Shtanko}\ \emph {et~al.}()\citenamefont {Shtanko},
  \citenamefont {Wang}, \citenamefont {Zhang}, \citenamefont {Harle},
  \citenamefont {Seif}, \citenamefont {Movassagh},\ and\ \citenamefont
  {Minev}}]{shtanko2023uncovering}%
  \BibitemOpen
  \bibfield  {author} {\bibinfo {author} {\bibfnamefont {O.}~\bibnamefont
  {Shtanko}}, \bibinfo {author} {\bibfnamefont {D.~S.}\ \bibnamefont {Wang}},
  \bibinfo {author} {\bibfnamefont {H.}~\bibnamefont {Zhang}}, \bibinfo
  {author} {\bibfnamefont {N.}~\bibnamefont {Harle}}, \bibinfo {author}
  {\bibfnamefont {A.}~\bibnamefont {Seif}}, \bibinfo {author} {\bibfnamefont
  {R.}~\bibnamefont {Movassagh}},\ and\ \bibinfo {author} {\bibfnamefont
  {Z.}~\bibnamefont {Minev}},\ }\href@noop {} {\bibinfo {title} {Uncovering
  local integrability in quantum many-body dynamics}},\ \Eprint
  {https://arxiv.org/abs/2307.07552} {arXiv:2307.07552} \BibitemShut {NoStop}%
\bibitem [{\citenamefont {Chen}\ \emph {et~al.}(2021)\citenamefont {Chen},
  \citenamefont {Yu}, \citenamefont {Zeng},\ and\ \citenamefont
  {Flammia}}]{Chen2020}%
  \BibitemOpen
  \bibfield  {author} {\bibinfo {author} {\bibfnamefont {S.}~\bibnamefont
  {Chen}}, \bibinfo {author} {\bibfnamefont {W.}~\bibnamefont {Yu}}, \bibinfo
  {author} {\bibfnamefont {P.}~\bibnamefont {Zeng}},\ and\ \bibinfo {author}
  {\bibfnamefont {S.~T.}\ \bibnamefont {Flammia}},\ }\bibfield  {title}
  {\bibinfo {title} {Robust shadow estimation},\ }\href
  {https://link.aps.org/doi/10.1103/PRXQuantum.2.030348} {\bibfield  {journal}
  {\bibinfo  {journal} {PRX Quantum}\ }\textbf {\bibinfo {volume} {2}},\
  \bibinfo {pages} {030348} (\bibinfo {year} {2021})}\BibitemShut {NoStop}%
\bibitem [{\citenamefont {Koh}\ and\ \citenamefont {Grewal}(2022)}]{Koh2020}%
  \BibitemOpen
  \bibfield  {author} {\bibinfo {author} {\bibfnamefont {D.~E.}\ \bibnamefont
  {Koh}}\ and\ \bibinfo {author} {\bibfnamefont {S.}~\bibnamefont {Grewal}},\
  }\bibfield  {title} {\bibinfo {title} {Classical shadows with noise},\ }\href
  {https://doi.org/10.22331/q-2022-08-16-776} {\bibfield  {journal} {\bibinfo
  {journal} {Quantum}\ }\textbf {\bibinfo {volume} {6}},\ \bibinfo {pages}
  {776} (\bibinfo {year} {2022})}\BibitemShut {NoStop}%
\bibitem [{\citenamefont {Vitale}\ \emph {et~al.}(2023)\citenamefont {Vitale},
  \citenamefont {Rath}, \citenamefont {Jurcevic}, \citenamefont {Elben},
  \citenamefont {Branciard},\ and\ \citenamefont
  {Vermersch}}]{vitale2023estimation}%
  \BibitemOpen
  \bibfield  {author} {\bibinfo {author} {\bibfnamefont {V.}~\bibnamefont
  {Vitale}}, \bibinfo {author} {\bibfnamefont {A.}~\bibnamefont {Rath}},
  \bibinfo {author} {\bibfnamefont {P.}~\bibnamefont {Jurcevic}}, \bibinfo
  {author} {\bibfnamefont {A.}~\bibnamefont {Elben}}, \bibinfo {author}
  {\bibfnamefont {C.}~\bibnamefont {Branciard}},\ and\ \bibinfo {author}
  {\bibfnamefont {B.}~\bibnamefont {Vermersch}},\ }\href@noop {} {\bibinfo
  {title} {Estimation of the quantum fisher information on a quantum
  processor}} (\bibinfo {year} {2023}),\ \Eprint
  {https://arxiv.org/abs/2307.16882} {arXiv:2307.16882 [quant-ph]} \BibitemShut
  {NoStop}%
\bibitem [{\citenamefont {LaBorde}\ and\ \citenamefont
  {Wilde}(2022)}]{LabordeWilde2022quantum}%
  \BibitemOpen
  \bibfield  {author} {\bibinfo {author} {\bibfnamefont {M.~L.}\ \bibnamefont
  {LaBorde}}\ and\ \bibinfo {author} {\bibfnamefont {M.~M.}\ \bibnamefont
  {Wilde}},\ }\bibfield  {title} {\bibinfo {title} {Quantum algorithms for
  testing hamiltonian symmetry},\ }\href
  {https://doi.org/10.1103/PhysRevLett.129.160503} {\bibfield  {journal}
  {\bibinfo  {journal} {Physical Review Letters}\ }\textbf {\bibinfo {volume}
  {129}},\ \bibinfo {pages} {160503} (\bibinfo {year} {2022})}\BibitemShut
  {NoStop}%
\bibitem [{\citenamefont {Kane}\ \emph {et~al.}(2017)\citenamefont {Kane},
  \citenamefont {Karmalkar},\ and\ \citenamefont {Price}}]{KaneEtAl2017robust}%
  \BibitemOpen
  \bibfield  {author} {\bibinfo {author} {\bibfnamefont {D.}~\bibnamefont
  {Kane}}, \bibinfo {author} {\bibfnamefont {S.}~\bibnamefont {Karmalkar}},\
  and\ \bibinfo {author} {\bibfnamefont {E.}~\bibnamefont {Price}},\ }\bibfield
   {title} {\bibinfo {title} {Robust polynomial regression up to the
  information theoretic limit},\ }in\ \href
  {https://doi.org/10.1109/FOCS.2017.43} {\emph {\bibinfo {booktitle} {2017
  IEEE 58th Annual Symposium on Foundations of Computer Science (FOCS)}}}\
  (\bibinfo {year} {2017})\ pp.\ \bibinfo {pages} {391--402}\BibitemShut
  {NoStop}%
\bibitem [{SM()}]{SM}%
  \BibitemOpen
  \href@noop {} {}\bibinfo {note} {{See Supplemental Material.}}\BibitemShut
  {Stop}%
\bibitem [{Note1()}]{Note1}%
  \BibitemOpen
  \bibinfo {note} {We use the asymptotic notation $\protect \tilde {\protect
  \mathcal {O}}(f(x))$ to denote $\protect \mathcal {O}(f(x)\protect \mathrm
  {polylog}(f(x)))$.}\BibitemShut {Stop}%
\bibitem [{\citenamefont {Beals}\ \emph {et~al.}(2001)\citenamefont {Beals},
  \citenamefont {Buhrman}, \citenamefont {Cleve}, \citenamefont {Mosca},\ and\
  \citenamefont {De~Wolf}}]{beals2001quantum}%
  \BibitemOpen
  \bibfield  {author} {\bibinfo {author} {\bibfnamefont {R.}~\bibnamefont
  {Beals}}, \bibinfo {author} {\bibfnamefont {H.}~\bibnamefont {Buhrman}},
  \bibinfo {author} {\bibfnamefont {R.}~\bibnamefont {Cleve}}, \bibinfo
  {author} {\bibfnamefont {M.}~\bibnamefont {Mosca}},\ and\ \bibinfo {author}
  {\bibfnamefont {R.}~\bibnamefont {De~Wolf}},\ }\bibfield  {title} {\bibinfo
  {title} {Quantum lower bounds by polynomials},\ }\href
  {https://dl.acm.org/doi/10.1145/502090.502097} {\bibfield  {journal}
  {\bibinfo  {journal} {Journal of the ACM (JACM)}\ }\textbf {\bibinfo {volume}
  {48}},\ \bibinfo {pages} {778} (\bibinfo {year} {2001})}\BibitemShut
  {NoStop}%
\bibitem [{\citenamefont {Pal}\ and\ \citenamefont
  {Huse}(2010)}]{PhysRevB.82.174411}%
  \BibitemOpen
  \bibfield  {author} {\bibinfo {author} {\bibfnamefont {A.}~\bibnamefont
  {Pal}}\ and\ \bibinfo {author} {\bibfnamefont {D.~A.}\ \bibnamefont {Huse}},\
  }\bibfield  {title} {\bibinfo {title} {Many-body localization phase
  transition},\ }\href {https://doi.org/10.1103/PhysRevB.82.174411} {\bibfield
  {journal} {\bibinfo  {journal} {Phys. Rev. B}\ }\textbf {\bibinfo {volume}
  {82}},\ \bibinfo {pages} {174411} (\bibinfo {year} {2010})}\BibitemShut
  {NoStop}%
\bibitem [{\citenamefont {{\v{Z}}nidari{\v{c}}}\ \emph
  {et~al.}(2008)\citenamefont {{\v{Z}}nidari{\v{c}}}, \citenamefont {Prosen},\
  and\ \citenamefont {Prelov{\v{s}}ek}}]{znidaric2008many}%
  \BibitemOpen
  \bibfield  {author} {\bibinfo {author} {\bibfnamefont {M.}~\bibnamefont
  {{\v{Z}}nidari{\v{c}}}}, \bibinfo {author} {\bibfnamefont {T.}~\bibnamefont
  {Prosen}},\ and\ \bibinfo {author} {\bibfnamefont {P.}~\bibnamefont
  {Prelov{\v{s}}ek}},\ }\bibfield  {title} {\bibinfo {title} {Many-body
  localization in the heisenberg xxz magnet in a random field},\ }\href
  {https://doi.org/10.1103%2Fphysrevb.77.064426} {\bibfield  {journal}
  {\bibinfo  {journal} {Physical Review B}\ }\textbf {\bibinfo {volume} {77}}
  (\bibinfo {year} {2008})}\BibitemShut {NoStop}%
\bibitem [{\citenamefont {Luitz}\ \emph {et~al.}(2015)\citenamefont {Luitz},
  \citenamefont {Laflorencie},\ and\ \citenamefont {Alet}}]{luitz2015many}%
  \BibitemOpen
  \bibfield  {author} {\bibinfo {author} {\bibfnamefont {D.~J.}\ \bibnamefont
  {Luitz}}, \bibinfo {author} {\bibfnamefont {N.}~\bibnamefont {Laflorencie}},\
  and\ \bibinfo {author} {\bibfnamefont {F.}~\bibnamefont {Alet}},\ }\bibfield
  {title} {\bibinfo {title} {Many-body localization edge in the random-field
  heisenberg chain},\ }\href {https://doi.org/10.1103/PhysRevB.91.081103}
  {\bibfield  {journal} {\bibinfo  {journal} {Physical Review B}\ }\textbf
  {\bibinfo {volume} {91}},\ \bibinfo {pages} {081103} (\bibinfo {year}
  {2015})}\BibitemShut {NoStop}%
\bibitem [{Note2()}]{Note2}%
  \BibitemOpen
  \bibinfo {note} {The precise transition disorder strength in the
  thermodynamic limit is subject to ongoing research, see e.g.~\cite
  {long2023phenomenology} and references therein.}\BibitemShut {Stop}%
\bibitem [{\citenamefont {Deutsch}(1991)}]{deutsch1991quantum}%
  \BibitemOpen
  \bibfield  {author} {\bibinfo {author} {\bibfnamefont {J.~M.}\ \bibnamefont
  {Deutsch}},\ }\bibfield  {title} {\bibinfo {title} {Quantum statistical
  mechanics in a closed system},\ }\href
  {https://link.aps.org/doi/10.1103/PhysRevA.43.2046} {\bibfield  {journal}
  {\bibinfo  {journal} {Phys. Rev. A}\ }\textbf {\bibinfo {volume} {43}},\
  \bibinfo {pages} {2046} (\bibinfo {year} {1991})}\BibitemShut {NoStop}%
\bibitem [{\citenamefont {Srednicki}(1994)}]{srednicki1994chaos}%
  \BibitemOpen
  \bibfield  {author} {\bibinfo {author} {\bibfnamefont {M.}~\bibnamefont
  {Srednicki}},\ }\bibfield  {title} {\bibinfo {title} {Chaos and quantum
  thermalization},\ }\href {https://doi.org/10.1103%2Fphysreve.50.888}
  {\bibfield  {journal} {\bibinfo  {journal} {Phys. Rev. E}\ }\textbf {\bibinfo
  {volume} {50}},\ \bibinfo {pages} {888} (\bibinfo {year} {1994})}\BibitemShut
  {NoStop}%
\bibitem [{Note3()}]{Note3}%
  \BibitemOpen
  \bibinfo {note} {In practice, we reconstruct all Pauli expectation values at
  a given time from the same experimental randomized measurement data. Thus,
  shot noise on different expectation values is in principle correlated. For
  large number of measurements, we however expect these correlation to be weak,
  and indeed confirm numerically that results obtained using the Gaussian noise
  approximation and actual randomized measurements are consistent.}\BibitemShut
  {Stop}%
\bibitem [{\citenamefont {Fisher}\ \emph {et~al.}(2023)\citenamefont {Fisher},
  \citenamefont {Khemani}, \citenamefont {Nahum},\ and\ \citenamefont
  {Vijay}}]{fisher2023random}%
  \BibitemOpen
  \bibfield  {author} {\bibinfo {author} {\bibfnamefont {M.~P.}\ \bibnamefont
  {Fisher}}, \bibinfo {author} {\bibfnamefont {V.}~\bibnamefont {Khemani}},
  \bibinfo {author} {\bibfnamefont {A.}~\bibnamefont {Nahum}},\ and\ \bibinfo
  {author} {\bibfnamefont {S.}~\bibnamefont {Vijay}},\ }\bibfield  {title}
  {\bibinfo {title} {Random quantum circuits},\ }\href
  {https://doi.org/10.1146/annurev-conmatphys-031720-030658} {\bibfield
  {journal} {\bibinfo  {journal} {Annual Review of Condensed Matter Physics}\
  }\textbf {\bibinfo {volume} {14}},\ \bibinfo {pages} {335} (\bibinfo {year}
  {2023})}\BibitemShut {NoStop}%
\bibitem [{\citenamefont {Albert}(2019)}]{albert2019asymptotics}%
  \BibitemOpen
  \bibfield  {author} {\bibinfo {author} {\bibfnamefont {V.~V.}\ \bibnamefont
  {Albert}},\ }\bibfield  {title} {\bibinfo {title} {Asymptotics of quantum
  channels: conserved quantities, an adiabatic limit, and matrix product
  states},\ }\href {https://doi.org/10.22331/q-2019-06-06-151} {\bibfield
  {journal} {\bibinfo  {journal} {Quantum}\ }\textbf {\bibinfo {volume} {3}},\
  \bibinfo {pages} {151} (\bibinfo {year} {2019})}\BibitemShut {NoStop}%
\bibitem [{\citenamefont {Cai}\ \emph {et~al.}(2023)\citenamefont {Cai},
  \citenamefont {Babbush}, \citenamefont {Benjamin}, \citenamefont {Endo},
  \citenamefont {Huggins}, \citenamefont {Li}, \citenamefont {McClean},\ and\
  \citenamefont {O'Brien}}]{cai2023quantum}%
  \BibitemOpen
  \bibfield  {author} {\bibinfo {author} {\bibfnamefont {Z.}~\bibnamefont
  {Cai}}, \bibinfo {author} {\bibfnamefont {R.}~\bibnamefont {Babbush}},
  \bibinfo {author} {\bibfnamefont {S.~C.}\ \bibnamefont {Benjamin}}, \bibinfo
  {author} {\bibfnamefont {S.}~\bibnamefont {Endo}}, \bibinfo {author}
  {\bibfnamefont {W.~J.}\ \bibnamefont {Huggins}}, \bibinfo {author}
  {\bibfnamefont {Y.}~\bibnamefont {Li}}, \bibinfo {author} {\bibfnamefont
  {J.~R.}\ \bibnamefont {McClean}},\ and\ \bibinfo {author} {\bibfnamefont
  {T.~E.}\ \bibnamefont {O'Brien}},\ }\href@noop {} {\bibinfo {title} {Quantum
  error mitigation}} (\bibinfo {year} {2023}),\ \Eprint
  {https://arxiv.org/abs/2210.00921} {arXiv:2210.00921 [quant-ph]} \BibitemShut
  {NoStop}%
\bibitem [{\citenamefont {Elben}\ \emph {et~al.}(2018)\citenamefont {Elben},
  \citenamefont {Vermersch}, \citenamefont {Dalmonte}, \citenamefont {Cirac},\
  and\ \citenamefont {Zoller}}]{elben2018renyi}%
  \BibitemOpen
  \bibfield  {author} {\bibinfo {author} {\bibfnamefont {A.}~\bibnamefont
  {Elben}}, \bibinfo {author} {\bibfnamefont {B.}~\bibnamefont {Vermersch}},
  \bibinfo {author} {\bibfnamefont {M.}~\bibnamefont {Dalmonte}}, \bibinfo
  {author} {\bibfnamefont {J.}~\bibnamefont {Cirac}},\ and\ \bibinfo {author}
  {\bibfnamefont {P.}~\bibnamefont {Zoller}},\ }\bibfield  {title} {\bibinfo
  {title} {R{\'{e}}nyi entropies from random quenches in atomic hubbard and
  spin models},\ }\href {https://doi.org/10.1103%2Fphysrevlett.120.050406}
  {\bibfield  {journal} {\bibinfo  {journal} {Physical Review Letters}\
  }\textbf {\bibinfo {volume} {120}} (\bibinfo {year} {2018})}\BibitemShut
  {NoStop}%
\bibitem [{\citenamefont {Bringewatt}\ \emph {et~al.}(2023)\citenamefont
  {Bringewatt}, \citenamefont {Kunjummen},\ and\ \citenamefont
  {Mueller}}]{bringewatt2023randomized}%
  \BibitemOpen
  \bibfield  {author} {\bibinfo {author} {\bibfnamefont {J.}~\bibnamefont
  {Bringewatt}}, \bibinfo {author} {\bibfnamefont {J.}~\bibnamefont
  {Kunjummen}},\ and\ \bibinfo {author} {\bibfnamefont {N.}~\bibnamefont
  {Mueller}},\ }\href@noop {} {\bibinfo {title} {Randomized measurement
  protocols for lattice gauge theories}} (\bibinfo {year} {2023}),\ \Eprint
  {https://arxiv.org/abs/2303.15519} {arXiv:2303.15519 [quant-ph]} \BibitemShut
  {NoStop}%
\bibitem [{\citenamefont {Long}\ \emph {et~al.}()\citenamefont {Long},
  \citenamefont {Crowley}, \citenamefont {Khemani},\ and\ \citenamefont
  {Chandran}}]{long2023phenomenology}%
  \BibitemOpen
  \bibfield  {author} {\bibinfo {author} {\bibfnamefont {D.~M.}\ \bibnamefont
  {Long}}, \bibinfo {author} {\bibfnamefont {P.~J.~D.}\ \bibnamefont
  {Crowley}}, \bibinfo {author} {\bibfnamefont {V.}~\bibnamefont {Khemani}},\
  and\ \bibinfo {author} {\bibfnamefont {A.}~\bibnamefont {Chandran}},\
  }\href@noop {} {\bibinfo {title} {Phenomenology of the prethermal many-body
  localized regime}},\ \Eprint {https://arxiv.org/abs/2207.05761}
  {arXiv:2207.05761 [cond-mat.dis-nn]} \BibitemShut {NoStop}%
\bibitem [{\citenamefont {Mirsky}(1960)}]{mirsky1960symmetric}%
  \BibitemOpen
  \bibfield  {author} {\bibinfo {author} {\bibfnamefont {L.}~\bibnamefont
  {Mirsky}},\ }\bibfield  {title} {\bibinfo {title} {Symmetric gauge functions
  and unitarily invariant norms},\ }\href@noop {} {\bibfield  {journal}
  {\bibinfo  {journal} {The quarterly journal of mathematics}\ }\textbf
  {\bibinfo {volume} {11}},\ \bibinfo {pages} {50} (\bibinfo {year}
  {1960})}\BibitemShut {NoStop}%
\bibitem [{\citenamefont {Pisier}(2003)}]{pisier2003introduction}%
  \BibitemOpen
  \bibfield  {author} {\bibinfo {author} {\bibfnamefont {G.}~\bibnamefont
  {Pisier}},\ }\href@noop {} {\emph {\bibinfo {title} {Introduction to operator
  space theory}}},\ Vol.\ \bibinfo {volume} {294}\ (\bibinfo  {publisher}
  {Cambridge University Press},\ \bibinfo {year} {2003})\BibitemShut {NoStop}%
\bibitem [{\citenamefont {Bandeira}\ \emph {et~al.}(2021)\citenamefont
  {Bandeira}, \citenamefont {Boedihardjo},\ and\ \citenamefont {van
  Handel}}]{bandeira2021matrix}%
  \BibitemOpen
  \bibfield  {author} {\bibinfo {author} {\bibfnamefont {A.~S.}\ \bibnamefont
  {Bandeira}}, \bibinfo {author} {\bibfnamefont {M.~T.}\ \bibnamefont
  {Boedihardjo}},\ and\ \bibinfo {author} {\bibfnamefont {R.}~\bibnamefont {van
  Handel}},\ }\bibfield  {title} {\bibinfo {title} {Matrix concentration
  inequalities and free probability},\ }\href@noop {} {\bibfield  {journal}
  {\bibinfo  {journal} {arXiv preprint arXiv:2108.06312}\ } (\bibinfo {year}
  {2021})}\BibitemShut {NoStop}%
\bibitem [{Note4()}]{Note4}%
  \BibitemOpen
  \bibinfo {note} {Here, $\protect \operatorname {OR}(\protect \bm {x})$ is
  defined as $\protect \operatorname {OR}(\protect \bm {x})=0$ if all $\protect
  \bm {x}_i$ are $0$ and $\protect \operatorname {OR}(\protect \bm {x})=1$
  otherwise.}\BibitemShut {Stop}%
\end{thebibliography}%

\clearpage
\onecolumngrid
\appendix

\begin{center}
{\large \textbf{  Supplemental Material: \vspace{0.1cm}\\ Learning conservation laws in unknown quantum dynamics}}
\end{center}

Our supplemental material is organized as follows. In Appendix~\ref{sec:counting_conservation_laws}, we prove Theorem~1 in the main text, providing a rigorous guarantee about the number of conserved quantities we learn. 
In Appendix~\ref{sec:accuracy_of_the_learned_conservation_laws}, we prove that all the conserved quantities are approximately contained in the subspace that we learn.
In Appendix~\ref{sec:test_conservation_laws_overall}, we prove Theorem~2 in the main text, which guarantees that we can test conservation laws for a single initial state efficiently with a high confidence level.  In Appendix~\ref{sec:ensemble_initial_states}, we prove Theorem~3 in the main text, showing that one can efficiently test whether a given observable is conserved on average for an ensemble of initial states.  In Appendix~\ref{sec:query_complexity_lower_bound}, we show that it is impossible to efficiently test whether an observable is conserved for all states in a black-box oracle setting.  In Appendix~\ref{sec:time_derivative_bounds}, we prove a technical lemma that is useful in the proof in Appendix~\ref{sec:test_conservation_laws_overall}.

\section{Counting conservation laws}
\label{sec:counting_conservation_laws}

In this section, we prove that the learning procedure outlined in Section II in the main text can reliably provide an upper bound for the number of conserved quantities. To simplify notation, we consider mostly the case of a single initial state, but comment on the appropriate re-definitions for the case of multiple initial states. The proof generalizes to the latter situation without requiring any change.

We define the expectation value matrix $X=(X_{ij})_{N_P\times N_T}$ (we require $N_T\geq N_P$) , where
\begin{equation}
    X_{ij} = \braket{P_i(t_j)}.
\end{equation}
From experiments we obtain estimates forming a matrix $\hat{X}=(\hat{X}_{ij})_{N_P\times N_T}$. The error is described by the following matrix
\begin{equation}
    E_{ij} = \hat{X}_{ij}-X_{ij}.
\end{equation}
The estimate is unbiased, which means $\mathbb{E}[E]=0$. The variance depends on what method we use to get the estimates $\hat{X}_{ij}$. In this section, we discuss two scenarios: the naive approach, which requires re-preparing the state for each expectation value, and classical shadows, enabling the simultaneous estimation of many expectation values.

The matrix $X$ can be used to characterize conservation laws. If an observable
\begin{equation}
    \sum_i c_i P_i
\end{equation}
is conserved, then we have
\begin{equation}
    \sum_i c_i \braket{P_i(t_j)} = \frac{1}{N_T}\sum_{ij'} c_i \braket{P_i(t_j')}
\end{equation}
for all $j=1,2,\cdots,N_T$. Writing the above equation in the matrix form, we have 
\begin{equation}
    c^T X = \frac{1}{N_T} c^T X \mathbf{1}\mathbf{1}^\top,
\end{equation}
where $\mathbf{1}=(1,1,\cdots,1)^{\top}$ is the $N_T$-dimensional vector with all entries being $1$.
Therefore, every conserved quantity is contained in the null space of the matrix
\begin{equation}
    W^{\top}=\left(I-\frac{1}{N_T} \mathbf{1}\mathbf{1}^\top\right)X^\top,
\end{equation}
where $I$ denotes the identity matrix. The dimension of its null space, which we denote by $D_{\mathrm{null}}$, provides an upper bound for the number of conserved quantities.

When we use multiple initial states, the relationship between $W$ and $X$ becomes slightly different. The entry-wise relationship is $W_{i,jk}=X_{i,jk}-N_T^{-1}\sum_{j'}X_{i,j'k}$, which we can reformulate in the matrix language as
\begin{equation}
\label{eq:from_X_to_W_multiple_initial_states}
    W^{\top} = \left[\left(I-\frac{1}{N_T} \mathbf{1}\mathbf{1}^\top\right)\otimes I\right]X^{\top},
\end{equation}
where $I-\frac{1}{N_T} \mathbf{1}\mathbf{1}^\top$ acts on the index $j$ and $I$ acts on the index $k$.

Note that $D_{\mathrm{null}}$ is not directly available to us. What we can do is to estimate $D_{\mathrm{null}}$ through the number of small singular values of the matrix
\begin{equation}
    \hat{W}^{\top}=\left(I-\frac{1}{N_T} \mathbf{1}\mathbf{1}^\top\right)\hat{X}^\top.
\end{equation}
When there are multiple initial states, the transformation from $\hat{X}$ to $\hat{W}$ is similar to \eqref{eq:from_X_to_W_multiple_initial_states}. Note that in both cases, because $\|I-\frac{1}{N_T} \mathbf{1}\mathbf{1}^\top\|\leq 1$, we have $\|W-\hat{W}\|\leq \|X-\hat{X}\|=\|E\|$.
We denote the number of singular values of the above matrix that are below $\epsilon$ by $\hat{D}_{\mathrm{null}}$. We show that, with enough samples, we can guarantee that with large probability
\begin{equation}
\label{eq:null_space_dimension_upper_bound}
    D_{\mathrm{null}}\leq \hat{D}_{\mathrm{null}}.
\end{equation}

We denote the singular values of $W$ and $\hat{W}$ by $\sigma_i$ and $\hat{\sigma}_i$, $j=1,2,\cdots,N_P$, arranged in ascending order, respectively. By Mirsky's inequality \cite{mirsky1960symmetric}, we have
\begin{equation}
    |\hat{\sigma}_i-\sigma_i|\leq \|W-\hat{W}\|\leq \|E\|,
\end{equation}
where $\|\cdot\|$ denotes the spectral norm. 
Consequently, if $\sigma_i=0$ for any $i$, then $\hat{\sigma}_i\leq \|E\|$. Therefore, we have the following lemma:
\begin{lem}
\label{lem:guarantee_through_E_norm}
When $\|E\|\leq \epsilon$, the inequality \eqref{eq:null_space_dimension_upper_bound} holds.
\end{lem}

Furthermore, we do not need the upper bound for $\|E\|$ to hold with probability 1. Rather, as long as the upper bound holds with probability that is greater than $1/2$ by a constant, we can simply repeat the procedure multiple times, and take the median of all $\hat{D}_{\mathrm{null}}$ that are computed. This  ensures that the median $\hat{D}_{\mathrm{null}}^{\mathrm{median}}\geq D_{\mathrm{null}}$ with probability at least $1-\delta$ with $\Or(\log(\delta^{-1}))$ using the Chernoff bound. Consequently, it suffices to upper bound $\mathbb{E}\|E\|$.
\begin{lem}
\label{lem:guarantee_through_E_norm_expectation}
When $\mathbb{E}\|E\|\leq \epsilon/4$, then
\begin{equation}
    D_{\mathrm{null}}\leq \hat{D}_{\mathrm{null}}^{\mathrm{median}},
\end{equation}
with probability at least $1-\delta$, where $\hat{D}_{\mathrm{null}}^{\mathrm{median}}$ is the median taken over $\Or(\log(\delta^{-1}))$ independent samples of $\hat{D}_{\mathrm{null}}$.
\end{lem}

Our main tool in bounding $\mathbb{E}\|E\|$ is through the noncommutative Khintchine inequality \cite[Section 9.8]{pisier2003introduction}. This inequality implies that, as stated in \cite[Eq. (1.2)]{bandeira2021matrix}, for a real symmetric random matrix $M=\sum_j g_j A_j$, where $g_j$'s are i.i.d. standard Gaussian, and $A_j$'s are real symmetric matrices of size $d\times d$, we have 
\begin{equation}
\label{eq:non_comm_khintchine}
    C_1 \sqrt{\Big\|\sum_j A_j^2\Big\|}\leq \mathbb{E}\|M\|\leq C_2 \sqrt{\Big\|\sum_j A_j^2\Big\|\log(d)}.
\end{equation}
Because of the requirement for $M$ to be real and symmetric, instead of directly considering $E$, we need to consider
\begin{equation}
    S_E = 
    \begin{pmatrix}
    0 & E \\
    E^{\top} & 0
    \end{pmatrix}
    =\sigma^-\otimes E + \sigma^+\otimes E^\top.
\end{equation}
Note that $\|S_E\|=\|E\|$. Therefore, we only need to upper bound $\|S_E\|$.

\subsection{The naive approach}
\label{sec:the_naive_approach}

Let us first consider the approach where each entry of $\hat{X}$ is sampled independently. In this scenario, the entries of matrix $E$ are independent. Therefore, we can write
\begin{equation}
\label{eq:E_decomposition_uncorrelated}
    E = \sum_{ij} g_{ij}\sigma_{ij}e_i e_j^\top.
\end{equation}
Here, $g_{ij}$ is a standard Gaussian random variable, and,
\begin{equation}
    \sigma_{ij}^2=\frac{\braket{P_i(t_j)^2}-\braket{P_i(t_j)}^2}{N_s},
\end{equation}
where $N_s$ is how many samples are used for each $\hat{X}_{ij}$. We are assuming that the error is Gaussian, which is reasonable for large $N_s$ due to the central limit theorem.

Correspondingly,
\begin{equation}
    S_E = \sum_{ij} g_{ij}\sigma_{ij}(\sigma^-\otimes e_i e_j^\top+\sigma^+\otimes e_j e_i^\top).
\end{equation}
By \eqref{eq:non_comm_khintchine} we have
\begin{equation}
\label{eq:upper_bound_SE_norm}
\begin{aligned}
    \mathbb{E}[\|S_E\|] &\leq C_2 \sqrt{\Big\|\sum_{ij} \sigma_{ij}^2 (\ket{0}\bra{0}\otimes e_i e_i^\top + \ket{1}\bra{1}\otimes e_j e_j^\top)\Big\|\log(N_P+N_T)} \\
    &\leq C_2 \sqrt{\max\Big\{\max_i\sum_j \sigma_{ij}^2,\max_j\sum_i \sigma_{ij}^2\Big\}\log(N_P+N_T)}.
\end{aligned}
\end{equation}
Because $\|P_i\|\leq 1$, we have
\begin{equation}
    \mathbb{E}[\|E\|]\leq \mathbb{E}[\|S_E\|]\leq \Or\left(\sqrt{\frac{N_P+N_T}{N_s}\log(N_P+N_T)}\right).
\end{equation}

To ensure $\mathbb{E}[\|E\|]\leq \epsilon/4$, we need to choose
\begin{equation}
    N_s = \wt{\Or}((N_P+N_T)\epsilon^{-2}).
\end{equation}
The total number of samples is therefore
\begin{equation}
\label{eq:num_samples_naive_approach}
    N_P\times N_T\times N_s\times \Or(\log(\delta^{-1})) = \wt{\Or}(N_P N_T (N_P+N_T)\epsilon^{-2}\log(\delta^{-1})).
\end{equation}
When considering $N_I$ initial states, we replace all $N_T$ with $N_T N_I$. The above expression then becomes
\begin{equation}
\label{eq:num_samples_naive_approach_multiple_initial_states}
    N_P\times N_T N_I\times N_s\times \Or(\log(\delta^{-1})) = \wt{\Or}(N_P N_T N_I (N_P+N_T N_I)\epsilon^{-2}\log(\delta^{-1})).
\end{equation}

Through Lemma \ref{lem:guarantee_through_E_norm_expectation}, we can compute the cost of learning 
\begin{thm}
\label{thm:naive_approach}
With $\wt{\Or}(N_P N_T N_I (N_P+N_T N_I)\epsilon^{-2}\log(\delta^{-1}))$ samples, we can ensure that $
    D_{\mathrm{null}}\leq \hat{D}_{\mathrm{null}}^{\mathrm{median}},
$
with probability at least $1-\delta$, where $\hat{D}_{\mathrm{null}}^{\mathrm{median}}$ is the median taken over $\Or(\log(\delta^{-1}))$ independent samples of $\hat{D}_{\mathrm{null}}$.
\end{thm}

Note that in practice, we usually do not choose $N_I$ to be large, but rather choose $N_T N_I=\Or(N)$ where $N$ is the system size.

\subsection{Using classical shadows}
\label{sec:using_classical_shadows}

Using classical shadows to construct $\hat{X}$, we no longer have the simple decomposition in \eqref{eq:E_decomposition_uncorrelated}. At each time $t_j$, the vector consisting of observable expectation values $\hat{X}_{\cdot j}$ is a Gaussian random vector with covariance matrix $\Sigma^j/N_s$ (again, Gaussianity is a result of the central limit theorem),
in which
\begin{equation}
\label{eq:covariance_classical_shadow}
    \Sigma^j_{ii'}=
    \begin{cases}
    3^{\omega(P_i,P_{i'})}\braket{P_i(t_j)P_{i'}(t_j)} - \braket{P_i(t_j)}\braket{P_{i'}(t_j)},&\text{ if }P_i\text{ and }P_{i'}\text{ completely commute}, \\
    -  \braket{P_i(t_j)}\braket{P_{i'}(t_j)},&\text{ otherwise},
    \end{cases}
\end{equation}
where by ``completely commute'' we mean that the two Pauli operators can be simultaneously diagonalized in the same single-qubit Pauli eigenbasis, and $\omega(P_i,P_{i'})$ is the number of qubits on which $P_i$ and $P_{i'}$ overlap.
Let us first perform an eigendecomposition for $\Sigma^j$:
\begin{equation}
    \Sigma^j = \sum_{l}\lambda^j_l v^j_l v^{j \top}_l,
\end{equation}
where $\lambda_l^j\geq 0$ because $\Sigma^j$ is symmetric positive semi-definite.
Then we have
\begin{equation}
    \hat{X}_{\cdot j} = \sum_{l} g_{lj}\sqrt{\lambda_l^j/N_s}v_l^j + X_{\cdot j},
\end{equation}
where ``$=$'' means equal in distribution, and $g_{lj}$'s are i.i.d. standard Gaussian. Then, the error matrix $E$ can be written as
\begin{equation}
    E = \sum_{lj} g_{lj}\sqrt{\lambda_l^j/N_s} v_l^j e_j^\top.
\end{equation}
Through the same analysis as in \eqref{eq:upper_bound_SE_norm}, we have
\begin{equation}
\label{eq:upper_bound_SE_norm_classical_shadow}
\begin{aligned}
    \mathbb{E}[\|S_E\|] 
    &\leq C_2 \sqrt{\max\Big\{\max_i\sum_j \frac{\lambda_l^{j}}{N_s},\max_j\sum_i \frac{\lambda_l^{j}}{N_s}\Big\}\log(N_P+N_T)} \\
    &\leq C_2\sqrt{\max_j\|\Sigma^j\|\frac{N_P+N_T}{N_s}\log(N_P+N_T)}
\end{aligned}
\end{equation}
Therefore
\begin{equation}
    \mathbb{E}[\|E\|]\leq \Or\left(\sqrt{\max_j\|\Sigma^j\|\frac{N_P+N_T}{N_s}\log(N_P+N_T)}\right).
\end{equation}

The next step is then to bound $\max_j\|\Sigma^j\|$. Note that in the worst case, we have $\max_j\|\Sigma^j\|=\Or(N_P)$. This is in fact attainable: we can choose $\rho(t_j)=\ket{\mathrm{GHZ}}\bra{\mathrm{GHZ}}$ where $\ket{\mathrm{GHZ}}=\frac{1}{\sqrt{2}}(\ket{00\cdots 0}+\ket{11\cdots 1})$, and let $P_i=Z_i$ for $i=1,2,\cdots,N_P$. Then, we have $\Sigma^j=2I+\mathbf{1}\mathbf{1}^\top$, thus giving us $\|\Sigma^j\|=\Or(N_P)$.

In this worst case, in order to ensure that $\mathbb{E}[\|E\|]\leq \epsilon/4$, we need
\begin{equation}
    N_s = \Or(N_P(N_P+N_T)\log(N_P+N_T)\epsilon^{-2}).
\end{equation}
The total number of samples needed is
\begin{equation}
    N_T\times N_s\times \Or(\log(\delta^{-1}))  = \wt{\Or}(N_P N_T (N_P+N_T)\epsilon^{-2}\log(\delta^{-1})).
\end{equation}
Note that here, even though we did not need to multiply by $N_P$ as in the naive approach, we still get the same sample complexity scaling as in \eqref{eq:num_samples_naive_approach}.

However, if correlation decays rapidly, classical shadows can offer us an advantage. More specifically, let us assume that the quantum system is on a $D$-dimensional lattice. Furthermore, for all $0\leq t\leq T$,
\begin{equation}
\label{eq:correlation_decay}
    |\braket{P_i(t)P_{i'}(t)} - \braket{P_i(t)}\braket{P_{i'}(t)}|\leq \Or(e^{-d(P_i,P_{i'})/\xi}),
\end{equation}
where $d(P_i,P_{i'})$ is the distance between $P_i$ and $P_{i'}$, and $\xi$ is the correlation length. Because $P_i$'s are supported on at most $k=\Or(1)$ adjacent qubits, the number of $P_{i'}$'s within $r$ distance from $P_i$ grows like $r^D$. Therefore we have
\begin{equation}
    \sum_{i'}|\braket{P_i(t)P_{i'}(t)} - \braket{P_i(t)}\braket{P_{i'}(t)}| \leq \Or(\xi^D),
\end{equation}
for all $i$. 
For $\Sigma^j_{ii'}$, in each row there are only $\Or(1)$ many entries that are different from $\braket{P_i(t)P_{i'}(t)} - \braket{P_i(t)}\braket{P_{i'}(t)}$, as can be seen from \eqref{eq:covariance_classical_shadow}. The absolute value of each entry is upper bounded by $\Or(1)$. Consequently
\begin{equation}
    \sum_{i'}\|\Sigma^j_{ii'}\|\leq \Or(\xi^D).
\end{equation}
Then we have
\begin{equation}
    \|\Sigma^j\|\leq \max_i \sum_{i'}\|\Sigma^j_{ii'}\|\leq \Or(\xi^D),
\end{equation}
which indicates $\|\Sigma^j\|=\Or(1)$ when $\xi,D=\Or(1)$.

In this good scenario, we only need
\begin{equation}
    N_s = \Or((N_P+N_T)\log(N_P+N_T)\epsilon^{-2}).
\end{equation}
The total number of samples needed is
\begin{equation}
    N_T\times N_s\times \Or(\log(\delta^{-1}))  = \wt{\Or}(N_T (N_P+N_T)\epsilon^{-2}\log(\delta^{-1})),
\end{equation}
which is quadratically better than the scaling in \eqref{eq:num_samples_naive_approach} in the $N_P$ dependence.
When we take into account having $N_I$ initial states, the number of samples then becomes
\begin{equation}
    N_T N_I\times N_s\times \Or(\log(\delta^{-1}))  = \wt{\Or}(N_T N_I (N_P+N_T N_I)\epsilon^{-2}\log(\delta^{-1})).
\end{equation}
Again, through Lemma \ref{lem:guarantee_through_E_norm_expectation}, we have
\begin{thm}
\label{thm:classical_shadow_approach}
We assume that the quantum system is defined on a $D$-dimensional lattice, and \eqref{eq:correlation_decay} holds for all $0\leq t\leq T$.
Then with $\wt{\Or}(N_T N_I (N_P+N_T N_I)\epsilon^{-2}\log(\delta^{-1}))$ samples, we can ensure that $
    D_{\mathrm{null}}\leq \hat{D}_{\mathrm{null}}^{\mathrm{median}},
$
with probability at least $1-\delta$, where $\hat{D}_{\mathrm{null}}^{\mathrm{median}}$ is the median taken over $\Or(\log(\delta^{-1}))$ independent samples of $\hat{D}_{\mathrm{null}}$.
\end{thm}

\section{Accuracy of the learned conservation laws}
\label{sec:accuracy_of_the_learned_conservation_laws}

In the learning procedure described in Section II in the main text, we obtain a subspace spanned by the singular vectors of the matrix $\hat{W}$. We show in this section that this subspace contains all the conserved quantities approximately.

We use the following result for singular vector perturbation to characterize the accuracy of the conservation laws obtained from our learning procedure:
\begin{lem}
\label{lem:singular_vector_perturbation}
    Suppose we have matrix $A\in\RR^{M\times N}$ and vector $w\in\RR^M$ such that $w^{\top}A=0$. 
    We let $\hat{A}=A+\delta A$. 
    We then write down the singular value decomposition of $\hat{A}$ as $\hat{A} = \sum_{k=1}^r \hat{\sigma}_k \hat{u}_k\hat{v}_k^{\top}$, where $\hat{\sigma}_1\leq\cdots\leq \hat{\sigma}_r$. Let $r'$ be the largest index such that $\hat{\sigma}_{r'}\leq \epsilon$.
    We decompose $w$ through
    \begin{equation}
    \label{eq:orthogonal_decomp_w}
        w = \hat{w} + w_{\perp},
    \end{equation}
    where $\hat{w} = \sum_{k=1}^{r'} \hat{u}_k\hat{u}_k^{\top} w$. Then we have
    \begin{equation}
        \|w_{\perp}\|\leq \frac{\|w^{\top}\delta A\|}{\epsilon}.
    \end{equation}
\end{lem}

\begin{proof}
    By the triangle inequality
    \begin{equation}
    \label{eq:triangle_ineq_wA}
        \|w^{\top}\hat{A}\| \leq \|w^{\top} A\| + \|w^{\top} \delta A\| = \|w^{\top} \delta A\|.
    \end{equation}
    On the other hand
    \begin{equation}
    \begin{aligned}
        \|w^{\top}\hat{A}\|^2 &= \hat{w}^{\top}\hat{A}\hat{A}^{\top}\hat{w} + \hat{w}^{\top}\hat{A}\hat{A}^{\top}w_{\perp} + w_{\perp}^{\top}\hat{A}\hat{A}^{\top}\hat{w} + w_{\perp}^{\top}\hat{A}\hat{A}^{\top}w_{\perp} \\
        &\geq w_{\perp}^{\top}\hat{A}\hat{A}^{\top}w_{\perp}.
    \end{aligned}
    \end{equation}
    Where we have used the fact that $\hat{w}^{\top}\hat{A}\hat{A}^{\top}\hat{w}\geq 0$, and
    \begin{equation}
        \hat{w}^{\top}\hat{A}\hat{A}^{\top}w_{\perp} = \sum_{k=1}^r \hat{\sigma}_k^2\hat{w}^{\top} \hat{u}_k\hat{u}_k^{\top} w_{\perp}=0
    \end{equation}
    because of the orthogonal decomposition \eqref{eq:orthogonal_decomp_w}. Similarly we have $w_{\perp}^{\top}\hat{A}\hat{A}^{\top}\hat{w}=0$. Because $w_{\perp}$ only overlaps with $\hat{u}_k$ for $k\geq r'+1$, we have
    \begin{equation}
        w_{\perp}^{\top}\hat{A}\hat{A}^{\top}w_{\perp} \geq \hat{\sigma}_{r'+1}^2 \|w_{\perp}\|^2\geq \epsilon^2 \|w_{\perp}\|^2.
    \end{equation}
    Therefore
    \begin{equation}
        \|w^{\top}\hat{A}\|^2 \geq  \epsilon^2 \|w_{\perp}\|^2.
    \end{equation}
    Combining the above with \eqref{eq:triangle_ineq_wA} we have
    \begin{equation}
        \|w_{\perp}\|\leq \frac{\|w^{\top}\delta A\|}{\epsilon}.
    \end{equation}
\end{proof}

In the context of our algorithm, we let $A=W$, $\hat{A}=\hat{W}$, and $w=\vec{c}$ ($\|\vec{c}\|=1$) in the above lemma. $\vec{c}$ here corresponds to an exact conserved quantity $O=\sum_i c_i P_i$, $\hat{W}$ is the shifted data matrix, in which we subtract the time average from all entries so that each row sums up to zero, and $W$ is its noiseless limit. 
In practice, $W$ is perturbed to be $\hat{W}$, and the corresponding perturbation is $\delta A = E(I-\mathbf{1}\mathbf{1}^{\top}/N_T)$, where $E$ contains the noise on each entry of the data matrix $X$.
The above result tells us that the subspace we obtain through performing SVD on $\hat{W}$  approximately contains the exact conservation law $O$ if
\begin{equation}
    \|\vec{c}^{\top}E(I-\mathbf{1}\mathbf{1}^{\top}/N_T)\|\ll \epsilon.
\end{equation}
Here, $\epsilon$ is our chosen truncation threshold for singular values. More precisely, the overlap between the vector $\vec{c}$ and the subspace spanned by $\hat{u}_k$, $k=1,2,\cdots,r'$, is at least
\begin{equation}
    \sqrt{1-\frac{\|\vec{c}^{\top}E(I-\mathbf{1}\mathbf{1}^{\top}/N_T)\|^2}{\epsilon}^2}\geq \sqrt{1-\|E\|^2/\epsilon^2}.
\end{equation}
This inequality also holds when we use multiple initial states through the relation \eqref{eq:from_X_to_W_multiple_initial_states}.

From the above analysis, we can see a tension in our choice of the threshold $\epsilon$: decreasing $\epsilon$  helps us better distinguish exactly conserved quantities from the approximate ones, but on the other hand, it increases the precision requirement on our data matrix $\hat{W}$.

\section{Testing conservation laws for a single initial state}
\label{sec:test_conservation_laws_overall}

In this section, we discuss how to test the conserved quantities that we have learned. The testing procedure is briefly outlined in Section~II in the main text. Here, we provide a more detailed description, prove its correctness, and also analyze the cost.

Let $f_i(t)=\Tr[\rho(t)O_i]$, where each $O_i$ is a sum of low-weight Pauli operators, for $i=1,2,\cdots,\chi$. We further assume that, using the notation $f^{(k)}(t)$ to denote the $k$th derivative of $f(t)$,
\begin{equation}
    |f_i^{(k)}(t)|\leq \mathcal{C}\Gamma^k k!.
\end{equation}
with constants $\mathcal{C}$ and $\Gamma$. We note that this is a very reasonable assumption to make. For time evolution under the von Neumann equation, this assumption holds with $\Gamma=\Or(1)$  for geometrically local Hamiltonians and Hamiltonians with certain fast-decaying long-range interaction. For details see Appendix~\ref{sec:time_derivative_bounds}. In general, we can always choose $\Gamma = \|H\|$. For the Lindblad master equation, a similar result can also be obtained.

\subsection{Expectation value interpolation for multiple observables}
\label{sec:exp_val_interpolation}

In this section, we find functions $\hat{p}_i(t)$, which are piecewise polynomials, such that
\begin{equation}
\label{eq:uniform_approx}
    |f_i(t)-\hat{p}_i(t)|\leq \epsilon
\end{equation}
with probability at least $1-\delta$ for each $t\in[0,T]$.

\subsubsection{Short-time interpolation}
\label{sec:short_time_interpolation}

We first propose a method for the case where $\Gamma T\leq 1$. In this case, $f(t)$ can be well-approximated by a polynomial through Taylor expansion. We have
\begin{equation}
    \begin{aligned}
        f_i(t) &= \sum_{k=0}^{\infty} \frac{f^{(k)}_i(T/2)}{k!}(t-T/2)^k \\
        &= \sum_{k=0}^K \frac{f^{(k)}_i(T/2)}{k!}(t-T/2)^k + \Or\left(\frac{\mathcal{C}}{2^K}\right),
    \end{aligned}
\end{equation}
where in deriving the second line, we have used the fact that
\begin{equation}
    \left|\frac{f^{(k)}_i(T/2)}{k!}(t-T/2)^k\right|\leq \mathcal{C}\left(\frac{\Gamma T}{2}\right)^k\leq \frac{\mathcal{C}}{2^k}. 
\end{equation}
We then denote
\begin{equation}
    p^{K}_i(t) = \sum_{k=0}^K \frac{f_i^{(k)}(T/2)}{k!}(t-T/2)^k.
\end{equation}
This is a degree $K$ polynomial satisfying
\begin{equation}
    \left|f_i(t)-p^{K}_i(t)\right|\leq \frac{\mathcal{C}}{2^K},
\end{equation}
for $t\in[0,T]$.

Following \cite{FrancaEtAl2022efficient}, which in turn relies on \cite{KaneEtAl2017robust}, we generate independent and identically distributed samples $t_1,t_2,\cdots, t_m$ from the Chebyshev distribution on $[0,T]$, specified by the probability density function $(1/\pi)(t(T-t))^{-1/2}$. We then generate $N_s$ classical shadows for each $t_j$, $j=1,2,\cdots,m$. Therefore we are able to generate estimates $y_{ij}$ such that
\begin{equation}
    \mathbb{E}[y_{ij}] = f_i(t_j),\quad \operatorname{var}[y_{ij}]=\Or\left(\frac{\|O_i\|_{\mathrm{shadow}}^2}{N_s}\right).
\end{equation}
Therefore, with an appropriately chosen constant factor, we have
\begin{equation}
    |y_{ij}-f_i(t_j)|\leq \Or\left(\frac{\|O_i\|_{\mathrm{shadow}}^2}{N_s}\right),
\end{equation}
with probability at least $2/3$. This indicates that
\begin{equation}
\label{eq:single_point_approx_err}
    |y_{ij}-p^{K}_i(t_j)|\leq \Or\left(\frac{\|O_i\|_{\mathrm{shadow}}^2}{N_s}\right) + \frac{\mathcal{C}}{2^K},
\end{equation}
with probability at least $2/3$. By the Chernoff-Hoeffding theorem, the above inequality holds for a majority of $j=1,2,\cdots,m$ with probability at least $1-e^{-\Omega(m)}$.

Using the robust polynomial interpolation method proposed in \cite{KaneEtAl2017robust}, and as stated in \cite[Theorem E.1]{FrancaEtAl2022efficient}, we can construct polynomials $\hat{p}_i(t)$ from $\{(t_j,y_{ij})\}$ such that
\begin{equation}
    |p^K_i(t)-\hat{p}_i(t)|\leq \Or\left(\frac{\|O_i\|_{\mathrm{shadow}}^2}{N_s}+ \frac{\mathcal{C}}{2^K}\right) 
\end{equation}
for all $t\in[0,T]$ with probability at least $1-\delta'$, by choosing
\begin{equation}
    \label{eq:m_condition}
    m = \Or\left(K\log(K\delta'^{-1})\right).
\end{equation}
This then leads to
\begin{equation}
\label{eq:uniform_approx_err}
    |f_i(t)-\hat{p}_i(t)|\leq \Or\left(\frac{\|O_i\|_{\mathrm{shadow}}^2}{N_s}+ \frac{\mathcal{C}}{2^K}\right) 
\end{equation}
for all $t\in[0,T]$. Therefore $\hat{p}_i(t)$ is a good uniform approximation of $f_i(t)$ for each $i$.

The above method can fail in two scenarios: either the error bound \eqref{eq:single_point_approx_err} fails to hold for a majority of times, or the sampled times fail to correctly capture the profile of the function. The former failure scenario has its probability bounded by $e^{-\Omega(m)}$ by the Chernoff-Hoeffding theorem, and the latter by $\delta'$ from the robust polynomial interpolation procedure. As a result, if we want to keep the total failure probability for a single $f_i(t)$ to be at most $\delta$, then we  only need
\begin{equation}
    e^{-\Omega(m)}+\delta'\leq \delta.
\end{equation}
To this end, and taking into account \eqref{eq:m_condition}, it suffices to choose
\begin{equation}
\label{eq:m_condition_intermediate}
    m = \Or\left(K\log(K\delta'^{-1})\right),\quad \delta'=\delta/2.
\end{equation}

We want the uniform approximation error in \eqref{eq:uniform_approx_err} to be upper bounded by $\epsilon$. Therefore we can choose
\begin{equation}
\label{eq:N_s_K_condition}
    N_s = \Or\left(\frac{\max_i\|O_i\|^2_{\mathrm{shadow}}}{\epsilon^2}\right),\quad K=\Or(\log(\mathcal{C}\epsilon^{-1})).
\end{equation}

Combining the above analysis, in particular \eqref{eq:m_condition_intermediate} and \eqref{eq:N_s_K_condition}, the total number of classical shadows we need is
\begin{equation}
    \label{eq:total_number_classical_shadows}
    N_s\times m = \Or\left(\frac{\max_i\|O_i\|^2_{\mathrm{shadow}}}{\epsilon^2}\log(\mathcal{C}\epsilon^{-1})\log\left(\frac{\log(\mathcal{C}\epsilon^{-1})}{\delta}\right)\right).
\end{equation}

We summarize the above analysis into the following lemma
\begin{lem}
    \label{lem:short_time_interpolation}
    Let $f_i(t)=\Tr[\rho(t)O_i]$, for $i=1,2,\cdots,\chi$. We further assume that $|f_i^{(k)}(t)|\leq \mathcal{C}\Gamma^k k!$. Then for $T\leq 1/\Gamma$ we can construct polynomials $\hat{p}_i(t)$, with degree up to $\Or(\log(\mathcal{C}\epsilon^{-1}))$, for $i=1,2,\cdots,\chi$, such that
    \begin{equation}
        \Pr\left[\max_{t\in[0,T]}|\hat{p}_i(t)-f_i(t)|>\epsilon\right]<\delta,
    \end{equation}
    using 
    \begin{equation}
        \Or\left(\frac{\max_i\|O_i\|^2_{\mathrm{shadow}}}{\epsilon^2}\log(\mathcal{C}\epsilon^{-1})\log\left(\frac{\log(\mathcal{C}\epsilon^{-1})}{\delta}\right)\right)
    \end{equation}
    classical shadows of the time evolved state $\rho(t)$.
\end{lem}

\subsubsection{Long-time interpolation}
\label{sec:long_time_interpolation}

We then consider the case where $T$ is not necessarily upper bounded by $1/\Gamma$. In this case we can simply partition the interval $[0,T]$ into segments each of length at most $1/\Gamma$, and there are therefore $\Gamma T$ such segments. We then use the algorithm described in Appendix~\ref{sec:short_time_interpolation} to generate a polynomial to approximate each $f_i(t)$ on each of the $\Gamma T$ segments. Piecing these polynomials together we have a piecewise polynomial approximation $\hat{g}_i(t)$ that approximates $f_i(t)$ for all $t\in[0,T]$. The success probability of this procedure can be obtained via a union bound. We therefore arrive a the following theorem from Lemma~\ref{lem:short_time_interpolation}:
\begin{thm}
    \label{thm:long_time_interpolation}
    Let $f_i(t)=\Tr[\rho(t)O_i]$, for $i=1,2,\cdots,\chi$. We further assume that $|f_i^{(k)}(t)|\leq \mathcal{C}\Gamma^k k!$. Then for $T>0$ we can construct piecewise-polynomial functions $\hat{g}_i(t)$, with degrees up to $\Or(\log(\mathcal{C}\epsilon^{-1}))$ on at most $\Gamma T$ segments, for $i=1,2,\cdots,\chi$, such that
    \begin{equation}
        \Pr\left[\max_{t\in[0,T]}|\hat{g}_i(t)-f_i(t)|>\epsilon\right]<\delta,
    \end{equation}
    using 
    \begin{equation}
        \Or\left(\frac{\Gamma T\max_i\|O_i\|^2_{\mathrm{shadow}}}{\epsilon^2}\log(\mathcal{C}\epsilon^{-1})\log\left(\frac{\Gamma T\log(\mathcal{C}\epsilon^{-1})}{\delta}\right)\right)
    \end{equation}
    classical shadows of the time evolved state $\rho(t)$.
\end{thm}

\subsection{Testing conservation laws}
\label{sec:test_conservation_laws}

For each $i=1,2,\cdots,\chi$, we define the time average $\bar{f}_i=(1/T)\int_0^T f_i(t) \dd t$. For each $i$, we want to test which of the two following hypotheses is true:
\begin{enumerate}
    \item[] \textbf{(Hypothesis 1)} For all $t\in[0,T]$, $f_i(t)=\bar{f}_i$;
    \item[] \textbf{(Hypothesis 2)} There exists $t^*\in[0,T]$ such that $|f_i(t^*)-\bar{f}_i|\geq \epsilon$.
\end{enumerate}

Unlike the usual statistical hypothesis testing situation, the two hypotheses we consider above are treated on an equal footing, and therefore we do not need to distinguish between the null hypothesis and the alternative hypothesis.  This is possible because we are considering a promise decision problem. 

With the piecewise polynomial approximations $\hat{g}_i(t)$ we have, we can easily distinguish the two cases: an $\epsilon/8$-uniform approximation ensures that 
\begin{equation}
    |\hat{g}_i(t)-f_i(t)|\leq \frac{\epsilon}{8},\quad \left|\frac{1}{T}\int_0^T \hat{g}_i(t) \dd t-\bar{f}_i\right|\leq \frac{\epsilon}{8}.
\end{equation}
Therefore if Hypothesis 1 is true, then we have
\begin{equation}
    \left|\frac{1}{T}\int_0^T \hat{g}_i(t) \dd t-\hat{g}_i(t)\right|\leq \frac{\epsilon}{8}+\frac{\epsilon}{8}=\frac{\epsilon}{4}.
\end{equation}
And if Hypothesis 2 is true, then we have
\begin{equation}
\begin{aligned}
    \left|\frac{1}{T}\int_0^T \hat{g}_i(t) \dd t-\hat{g}_i(t^*)\right| \geq |\bar{f}_i-f_i(t^*)|-\left|\frac{1}{T}\int_0^T \hat{g}_i(t) \dd t-\bar{f}_i\right| - \left|f_i(t^*)-\hat{g}_i(t^*)\right|\geq \frac{3\epsilon}{4}.
\end{aligned}
\end{equation}
Therefore the two hypotheses can be distinguished by the statistic $\max_{t\in[0,T]}\left|\frac{1}{T}\int_0^T \hat{g}_i(t) \dd t-\hat{g}_i(t)\right|\leq \epsilon/4$ or $\geq 3\epsilon/4$.

We can successfully distinguish between the two cases if we get an $\epsilon$-uniform approximation for $f_i(t)$. Therefore we can use Theorem~\ref{thm:long_time_interpolation} to determine the cost of keeping the error probability below $\delta$ (which means both the Type-I and Type-II error probabilities are below $\delta$). We therefore have the following theorem
\begin{thm}
    \label{thm:testing_conservation_laws_sm}
    Let $f_i(t)=\Tr[\rho(t)O_i]$, for $i=1,2,\cdots,\chi$. We further assume that $|f_i^{(k)}(t)|\leq \mathcal{C}\Gamma^k k!$. Then for $T>0$ we can distinguish between Hypotheses 1 and 2 for each $i$ correctly using 
    \begin{equation}
        \Or\left(\frac{\Gamma T\max_i\|O_i\|^2_{\mathrm{shadow}}}{\epsilon^2}\log(\mathcal{C}\epsilon^{-1})\log\left(\frac{\Gamma T\log(\mathcal{C}\epsilon^{-1})}{\delta}\right)\right)
    \end{equation}
    classical shadows of the time evolved state $\rho(t)$.
\end{thm}

\section{Testing conservation laws for an ensemble of initial states}
\label{sec:ensemble_initial_states}

In this section, we upper bound the generalization error given in Eq.~(9) in the main text, which we restate here
\begin{equation}
    \Big|\mathbb{E}_{\rho\sim\mathcal{D}}d(O_i,\rho)-\frac{1}{N_I}\sum_{k=1}^{N_I}d(O_i,\rho_k)\Big|.
\end{equation}
Note that $|d(O,\rho)|\leq 2\|O\|$ by the definition of $d(O,\rho)$ given in Eq.~(7). As a result, by Hoeffding's inequality, the above generalization error is at most $\epsilon'$ with probability at least 
\begin{equation}
    1-2\exp\left(-\frac{N_I\epsilon'^2}{2\|O\|^2}\right).
\end{equation}
This means that in order to make the generalization error to be at most $\epsilon$ with probability at least $1-\delta'$, we need
\begin{equation}
    N_I = \Or(\epsilon^{-2}\log(\delta'^{-1})\|O\|^2).
\end{equation}

\section{Query complexity lower bound for testing conservation laws for all initial states}
\label{sec:query_complexity_lower_bound}

In this section, we consider the setting where the Hamiltonian $e^{-iHt}$ on $N$ qubits is provided through an oracle, and we want to show that testing whether $H$ commutes with a simple observable can require $\Omega(2^{N/2})$ queries to the oracle in the worst case.

Our result is based on the lower bound for computing the OR function. In this setting, an $2^N$-bit string $\boldsymbol{x}=(\boldsymbol{x}_0,\boldsymbol{x}_1,\cdots,\boldsymbol{x}_{2^N-1})$ is provided through an oracle $U$ that satisfies
\begin{equation}
    \label{eq:oracle}
    U\ket{n} = (-1)^{\boldsymbol{x}_n}\ket{n},
\end{equation}
for $n=0,1,\cdots,2^N-1$.
In order to compute $\operatorname{OR}(\boldsymbol{x})$ \footnote{Here, $\operatorname{OR}(\boldsymbol{x})$ is defined as $\operatorname{OR}(\boldsymbol{x})=0$ if all $\boldsymbol{x}_i$ are $0$ and $\operatorname{OR}(\boldsymbol{x})=1$ otherwise.}, it is known that at least $\Omega(2^{N/2})$ queries to $U$ are needed \cite{beals2001quantum}. This query complexity lower bound still holds even if we constrain $\boldsymbol{x}$ to contain at most a single $1$, which corresponds to the partial function setting discussed in the comment after Theorem 4.13 in \cite{beals2001quantum}.

Now we choose  our Hamiltonian $H=U$ to be this oracle unitary  $U$, which is incidentally also Hermitian. We further restrict to the case where $\boldsymbol{x}$ contains at most a single $1$. We  first implement $e^{-iHt}$ using the oracle $U$ itself. Because the eigenvalues of $H=U$ are $\pm 1$, we can implement $e^{-iHt}$ using only two queries of $U$ for arbitrary $t$ through phase kickback.

We now assume that an algorithm can do the following: given access to $e^{-iHt}$ (acting on $N$ qubits) for arbitrarily chosen $t$, it can distinguish the following two cases
\begin{equation}
\label{eq:two_cases_commute}
    [H,X_1]=0,\text{ or }\|[H,X_1]\|\geq 1,
\end{equation}
where $X_1$ is the Pauli-X operator on the first qubit. If the algorithm can accomplish the task with $Q$ queries to $e^{-iHt}$, we  next argue that it can compute $\operatorname{OR}(\boldsymbol{x})$ using $2Q$ queries to $U$ (because of the implementation of $e^{-iHt}$ discussed in the previous paragraph), and in this way show that $Q=\Omega(2^{N/2})$.

Our argument goes as follows: if $\operatorname{OR}(\boldsymbol{x})=1$, then there exists $0\leq n^*\leq N-1$ such that $\boldsymbol{x}_{n^*}=1$, whereas $\boldsymbol{x}_{n}=0$ for all other $n$ because of our restriction of the domain of the OR function. Therefore
\begin{equation}
    H=U = I - 2\ket{n^*}\bra{n^*}.
\end{equation}
One can then compute
\begin{equation}
    \|[H,X_1]\|=2\|[\ket{n^*}\bra{n^*},X_1]\|=2,
\end{equation}
where we have used the fact that $X_1\ket{n^*}$ is orthogonal to $\ket{n^*}$. On the other hand, if $\operatorname{OR}(\boldsymbol{x})=0$, then $H=U=I$, and as a result $[H,X_1]=0$. Therefore, as long as we can distinguish the two cases in \eqref{eq:two_cases_commute}, we are able to evaluate $\operatorname{OR}(\boldsymbol{x})$. The above argument therefore leads us to the following theorem:
\begin{thm}
    \label{thm:lower_bound}
    Given access to $e^{-iHt}$ (acting on $N$ qubits) for arbitrarily chosen $t$ as a black-box oracle, any algorithm that can distinguish between $[H,X_1]=0$ and $\|[H,X_1]\|\geq 1$ with probability at least $2/3$, where $X_1$ is the Pauli-X operator on the first qubit, takes at least $\Omega(2^{N/2})$ queries to $e^{-iHt}$ in the worst case.
\end{thm}

\section{Time-derivative bounds for local observable expectation values}
\label{sec:time_derivative_bounds}

We consider a general Hamiltonian of the form
\begin{equation}
\label{eq:general_ham}
    H = \sum_{P\in\{I,X,Y,Z\}^N}\lambda_{P}P,
\end{equation}
where $|\lambda_{P}|\leq 1$, and $P=\bigotimes_j P_j$ is a Pauli operator with components $P_j\in\{I,X,Y,Z\}$. We show that the local observable expectation values behave nicely as a function of time for a class of Hamiltonians. More precisely, consider a local observable $O$, which by definition is supported on $s=\Or(1)$ adjacent qubits, we want to bound the high-order derivatives of
\begin{equation}
    \braket{O(t)}=\Tr[e^{iHt}Oe^{-iHt}\rho].
\end{equation}
The Hamiltonians $H$ we consider need to satisfy the following assumption:
\begin{lem}
\label{lem:derivative_growth}
Let $H$ in \eqref{eq:general_ham} be a $k$-local Hamiltonian. We assume that for each qubit $j$, $\sum_{P:P_j\neq I}|\lambda_P|\leq \Lambda$. Then
\begin{equation}
    \left|\frac{\dd^\ell}{\dd t^\ell}\braket{O(t)}\right|= \Or(\ell !(2\Lambda (k-1))^\ell \|O\|).
\end{equation}
\end{lem}

\begin{proof}
    First, we observe that
    \begin{equation}
        \frac{\dd^\ell}{\dd t^\ell}\braket{O(t)} = \tr[e^{iHt}\operatorname{ad}_H^{\ell}(O)e^{-iHt}\rho],
    \end{equation}
where we recall the notation that $\operatorname{ad}_A(B):=[A,B]$.
Therefore it suffices to prove that
\begin{equation}
    \|\operatorname{ad}_H^{\ell}(O)\|= \Or(\ell !(2\Lambda (k-1))^\ell \|O\|).
\end{equation}
We first expand $\operatorname{ad}_H^{\ell}(O)$ using the expression for $H$ in \eqref{eq:general_ham}:
\begin{equation}
\label{eq:nested_commutator_expand}
    \operatorname{ad}_H^{\ell}(O) = \sum_{P^1,P^2,\cdots,P^\ell}\lambda_{P^1}\lambda_{P^2}\cdots \lambda_{P^\ell}[P^{\ell},\cdots[P^2,[P^1,O]]\cdots].
\end{equation}
Note that the nested commutator $[P^{\ell},\cdots[P^2,[P^1,O]]\cdots]\neq 0$ only if each $P^r$ overlaps with the nested commutator up to the $r-1$ level, $r=1,2,\cdots,\ell$. For a fixed sequence of $\{P^r\}$, we recursively define these nested commutators through
\begin{equation}
    O_0 = O,\quad O_{r} = [P^r,O_{r-1}],\ r=1,2\cdots,\ell.
\end{equation}
Then for each $P^r$, there must exist a qubit $q_r$, such that both $P^r$ and $O_{r-1}$ act non-trivially on $q_r$, in order for $[P^{\ell},\cdots[P^2,[P^1,O]]\cdots]\neq 0$. Consequently we have
\begin{equation}
\label{eq:huge_sum_paths}
    \begin{aligned}
        & \sum_{P^1,P^2,\cdots,P^\ell}|\lambda_{P^1}\lambda_{P^2}\cdots \lambda_{P^\ell}|\|[P^{\ell},\cdots[P^2,[P^1,O]]\cdots]\| \\
        &\leq \sum_{q_1\in\operatorname{supp(O)}}\sum_{P^1:P^1_{q_1}\neq I}\sum_{q_2\in\operatorname{supp(O_1)}}\sum_{P^2:P^2_{q_2}\neq I}\cdots \sum_{q_\ell\in\operatorname{supp(O_{\ell-1})}}\sum_{P^\ell:P^\ell_{q_\ell}\neq I}|\lambda_{P^1}\lambda_{P^2}\cdots \lambda_{P^\ell}|\|O\| \\
        &=\|O\|\sum_{q_1\in\operatorname{supp(O)}}\sum_{P^1:P^1_{q_1}\neq I}|\lambda_{P^1}|\sum_{q_2\in\operatorname{supp(O_1)}}\sum_{P^2:P^2_{q_2}\neq I}|\lambda_{P^2}|\cdots \sum_{q_\ell\in\operatorname{supp(O_{\ell-1})}}\sum_{P^\ell:P^\ell_{q_\ell}\neq I} |\lambda_{P^\ell}| \\
    \end{aligned}
\end{equation}
Note that $\sum_{P^\ell:P^\ell_{q_\ell}\neq I} |\lambda_{P^\ell}|\leq \Lambda$ by assumption. As a result
\begin{equation}
    \sum_{q_\ell\in\operatorname{supp(O_{\ell-1})}}\sum_{P^\ell:P^\ell_{q_\ell}\neq I} |\lambda_{P^\ell}|\leq |\operatorname{supp}(O_{\ell-1})|\Lambda\leq (s+(\ell-1)(k-1))\Lambda,
\end{equation}
where we have used the fact that $|\operatorname{supp}(O_{\ell-1})|\leq s+(\ell-1)(k-1)$, which can be proved by induction on $\ell$. Because of this, the right-hand side of \eqref{eq:huge_sum_paths} can be upper bounded by
\begin{equation}
    \begin{aligned}
        &(s+(\ell-1)(k-1))\Lambda \\
        &\times 
        \|O\|\sum_{q_1\in\operatorname{supp(O)}}\sum_{P^1:P^1_{q_1}\neq I}|\lambda_{P^1}|\sum_{q_2\in\operatorname{supp(O_1)}}\sum_{P^2:P^2_{q_2}\neq I}|\lambda_{P^2}|\cdots \sum_{q_{\ell-1}\in\operatorname{supp(O_{\ell-2})}}\sum_{P^\ell:P^{\ell-1}_{q_{\ell-1}}\neq I} |\lambda_{P^{\ell-1}}|.
    \end{aligned}
\end{equation}
One can keep doing this for $\ell$ times, and the right-hand side of \eqref{eq:huge_sum_paths} is bounded by
\begin{equation}
    \|O\|\Lambda^\ell s(s+k-1)\cdots(s+(\ell-1)(k-1)).
\end{equation}
Because of \eqref{eq:nested_commutator_expand}, this is an upper bound of $\|\operatorname{ad}_H^{\ell}(O)\|$. 

We only need to bound $s(s+k-1)\cdots(s+(\ell-1)(k-1))$.
We have
\begin{equation}
\begin{aligned}
    s(s+k-1)\cdots(s+(\ell-1)(k-1)) &= (k-1)^{\ell}\frac{s}{k-1}\left(\frac{s}{k-1}+1\right)\cdots \left(\frac{s}{k-1}+\ell-1\right) \\
    &\leq (k-1)^{\ell} \frac{\left(\lceil\frac{s}{k-1}\rceil + \ell-1 \right)!}{\lceil\frac{s}{k-1}-1\rceil!} \\
    &=(k-1)^{\ell} {\lceil\frac{s}{k-1}\rceil + \ell-1 \choose \ell}\ell ! \\
    &\leq (k-1)^{\ell} 2^{\lceil\frac{s}{k-1}\rceil + \ell-1} \ell ! \\
    &= \Or((2(k-1))^{\ell}\ell !).
\end{aligned}
\end{equation}
This completes the proof.
\end{proof}

Next, we show that geometrically local Hamiltonians and local Hamiltonians with power-law interaction that decays fast enough satisfy the assumptions in Lemma~\ref{lem:derivative_growth}. The key quantity of interest is $\sum_{P:P_j\neq I}|\lambda_P|$, which is the sum of the absolute value of all coefficients of Pauli terms that act non-trivially on a qubit $j$. For geometrically local Hamiltonians, there are only $\Or(1)$ terms acting on any given qubit, and consequently $\sum_{P:P_j\neq I}|\lambda_P|=\Or(1)$ if $|\lambda_P|\leq 1$ as assumed at the beginning of this section.

For power-law interaction Hamiltonians, we adopt a restricted definition to make the discussion easier, without neglecting any essential feature of these Hamiltonians. For these Hamiltonians, $\lambda_P\neq 0$ only when $P$ involves at most two qubits. Moreover, $|\lambda_P|=\Or(d^{-\alpha})$, where $d$ is the distance, on a $D$-dimensional lattice, between the two qubits, and $\alpha$ is the exponent deciding how rapid the decay is. The sum of all coefficients involving a qubit $j$ can be roughly bounded by
\begin{equation}
    \sum_{j'\in\ZZ^D} |j'|^{-\alpha},
\end{equation}
where $Z$ is the set of all integers, and $\ZZ^D$ is a $D$-dimensional lattice.
When $\alpha>D+1$, we have $\sum_{j'\in\ZZ^D} |j'|^{-\alpha}<\infty$, thus giving us a bound $\sum_{P:P_j\neq I}|\lambda_P|=\Or(1)$.

\end{document}